\documentclass[lettersize,journal]{IEEEtran}
\usepackage{xcolor}%
\usepackage{enumitem}%
\usepackage{amssymb}%
\usepackage{amsthm}
\newtheorem{theorem}{Theorem}
\usepackage{amsmath,amsfonts}
\usepackage{algorithmic}
\usepackage{algorithm}
\usepackage{array}
\usepackage{subfigure}
\usepackage{textcomp}
\usepackage{stfloats}
\usepackage{url}
\usepackage{verbatim}
\usepackage{graphicx}
\usepackage{cite}
\usepackage{diagbox}
\usepackage{multirow}
\allowdisplaybreaks
\begin{document}

\title{Optimization and Deep Learning based Resource Allocation for UAV-Aided Wireless Communication with Rotatable Antenna Array\vspace{-0.2cm}}

\author{Fengcheng Pei, Lin Xiang, Anja Klein, and Robert Schober
        % <-this % stops a space
\vspace{-1.1cm}        }

\maketitle

\begin{abstract}
Multi-antenna unmanned aerial vehicle (UAV)-aided communication presents a promising solution to increase the system capacity and improve the quality of service (QoS) of the future wireless networks. In this paper, we equip a UAV platform with a rotatable antenna array (RAA), which can be rotated flexibly in three-dimensional (3D) space via an onboard gimbal, enabling additional spatial degrees of freedom (DoFs) for improving multiuser signal transmission and interference management. Compared with a conventional fixed antenna array (FAA), the RAA can proactively align users with the high-gain region of its antenna elements and reduce the spatial channel correlations among users. To demonstrate the advantages of RAA, we jointly design the RAA orientation and beamforming to maximize the sum-rate of multiple users subject to per-user QoS constraints. The formulated problem is highly nonconvex and exhibits strong coupling between the RAA orientation and beamforming variables. To solve this challenging problem, we propose first an optimization framework based on the penalty dual decomposition (PDD) method to iteratively optimize RAA orientation and beamforming. While the optimization framework yields high reliability in QoS satisfaction and favorable sum-rate performance, its iterative nature may hinder real-time deployment. To accelerate the joint design and preserve a high-quality solution, we further propose a deep learning (DL) framework based on graph neural networks (GNNs). Simulation results demonstrate that RAAs significantly outperform FAAs in UAV-aided communication. Additionally, the proposed optimization framework is capable of satisfying stringent QoS requirements with high reliability, while the proposed DL framework attains comparable sum-rate performance with substantially reduced computation time and exhibits robustness to user position information errors.
\end{abstract}
\vspace{-0.2cm}
\vspace{-0.8cm}
\section{Introduction}
Multi-antenna unmanned aerial vehicle (UAV)-aided wireless communication presents a promising solution to increase the spectral efficiency and enhance the service capabilities of sixth-generation (6G) cellular networks in both normal and emergency scenarios \cite{r1, r2,r3}. However, the elevated altitude of UAV platforms gives rise to line-of-sight (LoS) dominant channels for terrestrial users, which may create severe multiuser interference and hinder support for high user densities. Furthermore, UAV platforms are typically constrained by size, weight, and power (SWAP), restricting the dimensions of onboard antenna arrays, which further complicates multiuser interference management.

To address these challenges, existing research has primarily investigated the joint design of beamforming and UAV trajectory \cite{r5, r6} to maximize, e.g., the sum-rate, where the orientations of the onboard antenna arrays are fixed. For example, the UAV in \cite{r5} is equipped with a vertically placed uniform linear array (ULA) throughout flight. Such \emph{fixed antenna arrays (FAAs)} limit the ability to spatially separate users in the angular domain and may lead to highly correlated channels among users, significantly restricting the sum-rate performance. Moreover, during the flight of the UAV, users may fall within the low-gain regions of the FAA antenna elements, jeopardizing signal transmission and connectivity.

To overcome the drawbacks of FAAs, \emph{rotatable antenna arrays (RAAs)} have been proposed in \cite{r10}. Unlike FAAs, the orientation of RAAs can be adjusted and optimized flexibly in three-dimensional (3D) space, providing additional spatial degrees of freedom (DoFs). Through orientation design, the RAAs can separate users in the angular domain to mitigate multiuser interference and align users with the high-gain regions of their antenna elements to enhance channel gains. Compared with other intelligent antenna technologies, such as movable antennas (MAs) \cite{r7}, element-wise rotatable antennas (RAs) \cite{r8}, or six-dimensional MAs (6DMAs) \cite{6DMA}, which require complex mechanical structures and antenna feeders to reposition \cite{r7} or rotate \cite{r8} individual antenna elements, or adjust both the positions and orientations of 6DMA surfaces \cite{6DMA}, the antenna elements within an RAA remain static relative to each other during RAA rotation, which can be easily realized by a lightweight micro gimbal mounted on a UAV. Therefore, the simple structure of RAAs offers a favorable trade-off between exploiting spatial DoFs and implementation complexity for UAV platforms constrained by SWAP \cite{r10}. Additionally, unlike the UAV trajectory design in \cite{r5,r6}, which requires continuous repositioning of the UAV, RAAs can also provide extra spatial DoFs for a hovering UAV.

The existing literature has investigated UAV-aided communication with RAAs from various perspectives, including transmit power minimization \cite{r11}, spectral efficiency maximization \cite{r14}, energy efficiency maximization \cite{r12}, completion time minimization \cite{r13}, as well as optimization of integrated sensing and communication (ISAC) \cite{r15,r16}. Specifically, the authors of \cite{r11} minimize transmit power by jointly optimizing a hybrid precoder and RAA orientation in UAV-aided massive multiple-input multiple-output (MIMO) systems. In\cite{r14}, the authors investigate spectral efficiency maximization via joint transmit beamforming and array orientation design. Energy efficiency is maximized in \cite{r12} by jointly accounting for the total power consumption associated with UAV hovering, communication, and RAA rotation. The authors in \cite{r13} focus on minimizing the communication mission completion time of the UAV through the joint design of user scheduling, RAA orientation, and beamforming. Finally, UAV-aided ISAC systems with RAAs are investigated in \cite{r15,r16}. 

The existing results in \cite{r11,r12,r13,r14,r15,r16} have demonstrated significant performance gains of RAAs over FAAs. However, they focus on either maximizing the overall system performance without considering user fairness \cite{r12,r14} or minimizing the  required resource such as transmit power and task completion time to rigidly satisfy predefined quality-of-service (QoS) requirements \cite{r11,r13}. In the former case, the users with poor channels may become inactive in order to maximize the sum spectral efficiency, whereas in the latter case, the system only strives to satisfy the QoS requirements with equality, restricting the flexibility in resource allocation. As such, the potential of RAAs for intelligently and flexibly balancing system throughput and user fairness, which is central to practical deployment and resource allocation, remains uninvestigated. This research gap motivates us to jointly design RAA orientation and beamforming for maximization of the sum-rate subject to per-user minimum achievable rate requirements for UAV-aided communication. In this scenario, the RAA is particularly beneficial, as it can not only proactively improve the channel conditions of weak users through orientation design. Also, exploiting the flexible resource allocation in the considered joint design, the RAA can enable the system to achieve even higher system throughput or meet more stringent QoS requirements.

Nevertheless, the formulated joint design problem is highly nonconvex, with strong coupling between the RAA orientation and beamforming variables in both the objective function and QoS constraints, which makes the conventional block coordinate descent (BCD) methods as adopted in \cite{r12,r13,r14,r15,r16} prone to converging to poor local optima. To tackle this challenge, we propose first an optimization framework based on the penalty dual decomposition (PDD) method \cite{r25,r26}, which decouples the variables in the constraints before iteratively optimizing the RAA orientation and beamforming. The PDD method is guaranteed to converge to a Karush–Kuhn–Tucker (KKT) point. However, it typically requires a large number of iterations \cite{r27} for convergence, hindering real-time decision-making. Moreover, accurate user position information is crucial for designing RAA orientation and beamforming based on optimization methods \cite{r11,r12,r13,r14,r15,r16}, and the authors of \cite{r11} point out that even small position errors can dramatically degrade system performance.

These considerations inspire us to further adopt deep learning (DL) techniques to solve the joint design problem. Particularly, graph neural networks (GNNs) are powerful tools for resource allocation in wireless communications as they can leverage the underlying graph topology of wireless networks \cite{r17,r18,r19}. For example, the authors of \cite{r19} employ GNNs for beamforming design in an FAA-based multiuser MIMO system, and simulation results show that GNNs approach the performance of conventional optimization methods with greatly reduced computation time. Additionally, DL techniques empirically exhibit robustness against imperfect context or side information by learning input-output mappings\cite{r20}. 

Despite the advantages of DL and GNNs, directly applying the existing architectures in \cite{r17,r18,r19} to our problem leads to severe gradient imbalance during network training. This is because the RAA orientation and beamforming variables are strongly coupled and their gradients differ by orders of magnitude due to their distinct physical scales. Thus, a tailored GNN-based architecture and training strategy need to be developed. Motivated by this research gap, we propose a DL framework consisting of two key components: A GNN-based architecture with two serially connected modules dedicated to the RAA orientation and beamforming design, respectively, and a two-stage training strategy. We further conduct a comprehensive comparison between the two proposed design frameworks via simulations, covering the sum-rate, QoS satisfaction, computation time, and robustness to user position errors. To the best of our knowledge, DL methods have not yet been investigated for RAA-based communication systems. Moreover, while both optimization and DL are appealing approaches for the design of RAAs, a thorough comparison between them is lacking. Our main contributions can be summarized as follows:
\begin{itemize}[leftmargin=*, itemsep=0pt, topsep=0pt, parsep=0pt]
    \item We investigate an RAA-based UAV-aided multiuser downlink communication system, and we formulate the challenging problem of jointly optimizing RAA orientation and beamforming, aiming to maximize the sum-rate subject to per-user QoS constraints.
    \item The joint design problem is highly nonconvex and the variables are strongly coupled. To solve this problem, we propose an optimization framework based on the PDD method to decouple the variables and iteratively optimize them, with closed-form solutions derived for the subproblems. The proposed optimization framework is guaranteed to converge to a KKT point under mild conditions.
    \item We further propose a DL framework for the joint design problem with unsupervised learning. To overcome the gradient imbalance issue, we develop a GNN-based architecture with two serially connected modules, which output the RAA orientation and beamforming vectors, respectively, and design a two-stage training strategy for the two modules. The proposed DL framework satisfies scalability and input permutation equivariance. 
    \item Simulation results demonstrate that RAA-based UAV systems significantly outperform FAA-based ones in terms of sum-rate and QoS satisfaction. Additionally, the optimization and DL frameworks show complementary advantages. Specifically, the DL framework approaches the performance of the optimization framework with substantially reduced computation time and exhibits robustness to imperfect user position information, while the optimization framework can more reliably satisfy stringent QoS requirements than the DL framework.
\end{itemize} 

\emph{Notations:} $\Re(x)$ and $\Im(x)$ represent the real and imaginary parts of a complex number $x$, respectively. %$\mathbb{C}^{m \times n}$ and $\mathbb{R}^{m \times n}$ denote $m \times n$ complex- and real-valued matrices, respectively. Then, 
The $(m,n)$-th entry, the $m$-th row, and the $n$-th column of a matrix $\boldsymbol{A}$ are denoted by $\boldsymbol{A}_{(m,n)}$, $\boldsymbol{A}_{(m,:)}$, and $\boldsymbol{A}_{(:,n)}$, respectively. The trace of a square matrix $\boldsymbol{A}$ is given by tr$\{\boldsymbol{A}\}$. Furthermore, $(\cdot)^{\star}$ represents the optimal value of an optimization variable. Finally, $\boldsymbol{a}\mathbin{/\mkern-5mu/}\boldsymbol{b}$ represents the concatenation of vectors $\boldsymbol{a}$ and $\boldsymbol{b}$, and $\mathbin{/\mkern-5mu/}_{k=1}^{K}\boldsymbol{x}_k$ denotes the successive concatenation of vectors $\boldsymbol{x}_k$, for $k=1,2,...,K$.
\vspace{-0.3cm}
\section{System Model and Problem Formulation}
\vspace{-0.15cm}
\begin{figure}[t]
\centering \includegraphics[width=0.95\columnwidth]{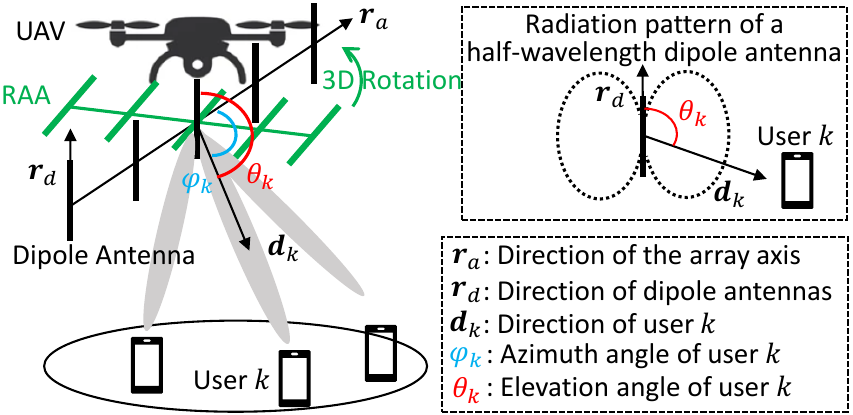} 
\vspace{-0.15cm}
\caption{UAV-aided communication system equipped with an RAA.}
\label{System}
\end{figure}
As shown in Fig.~\ref{System}, we consider downlink multiuser communication employing a multi-antenna rotary-wing UAV. The UAV hovers in the air and serves as a BS to communicate with $K$ ground users. Unlike the commonly adopted idealized isotropic antenna model, we consider a practical half-wavelength dipole antenna model for signal transmission. Due to the elevated altitude of the UAV, signal transmission will be dominated by the LoS channels \cite{r1,r2}. In this case, when FAAs composed of dipole antennas are employed, a large number of users may lead to severe LoS multiuser interference due to high channel correlations, and some users may also fall into low-gain regions of the dipole antennas. To overcome this predicament, we equip the UAV with an RAA composed of $N$ uniformly linearly spaced dipole antennas, with $N\geq K$. As depicted in Fig.~\ref{System}, the RAA is mounted on the UAV via a micro gimbal, enabling flexible 3D rotation around its centroid. By orientation design, the RAA can spatially separate users in the angular domain to mitigate channel correlation and also align them with the high-gain regions of dipole antennas, thereby enhancing overall system performance. We assume that each communication user has an isotropic receive antenna.
\vspace{-0.6cm}
\subsection{3D Channel Model}
We adopt a 3D Cartesian coordinate system, where the ground corresponds to the $xOy$ plane and the UAV is assumed to hover in the air at a given position $\boldsymbol{p}=[p_x,p_y,p_z]^\mathrm{T}$, with $p_z>0$. Let $\mathcal{K}\triangleq\{1,...,K\}$ be the set of users, where the position of user $k\in \mathcal{K}$ is denoted by $\boldsymbol{u}_k=[u_{k,x},u_{k,y},0]^\mathrm{T}$. We define unit vector $\boldsymbol{d}_k\!\triangleq\!(\boldsymbol{u}_k\!-\!\boldsymbol{p})/D_k\!\in\!\mathbb{R}^{3\times1}$ as the direction of user $k$ as seen from the UAV, where $D_k\!=\!\| \boldsymbol{u}_k\!-\!\boldsymbol{p} \|$ is the distance between the UAV and user $k$. Moreover, as illustrated in Fig.~\ref{System}, we define two orthonormal vectors $\boldsymbol{r}\!_{\text{a}}\!\in\!\mathbb{R}^{3\times1}$ and $\boldsymbol{r}\!_{\text{d}}\!\in\!\mathbb{R}^{3\times1}$ to represent the directions of the array axis and the dipoles in the RAA, respectively, which together characterize the RAA orientation \cite{r11,r12,r13,r14}.

We suppose that the UAV knows each user $k$'s position a priori and all users are assumed to experience LoS channels \cite{r11,r12,r13,r14}. Thus, the UAV-to-user $k$ channel $\boldsymbol{h}_k\in\mathbb{C}^{N\times1}$ is modeled as 
\begin{align}
\boldsymbol{h}_{k}=\sqrt{\beta_k}\cdot\boldsymbol{a}_{k},
\label{channel}
\end{align}
where $\beta_k=10^{-(140.7+36.7\log_{10}\frac{D_k}{1000})/10}$ is the path loss for $D_k$ in meter \cite{r21}.  $\boldsymbol{a}_{k}\in\mathbb{C}^{N\times1}$ denotes the transmit steering vector for user $k$ and is given by 
\begin{align}
\boldsymbol{a}_{k}=\alpha\cdot e_k \cdot \tilde{\boldsymbol{a}}_k,
\label{steering_vector}
\end{align}
where $e_k\in\mathbb{R}$ and $\tilde{\boldsymbol{a}}_k\in\mathbb{C}^{N\times1}$ represent the element factor of dipole antennas and the array factor of the RAA, respectively \cite{r22}.
As depicted in Fig.~\ref{System}, we refer to $\theta_k\in [0,\pi]$ and $\varphi_k\in [0,\pi]$ as the elevation and azimuth angles of user $k$ seen from the RAA, respectively, and it is immediate to show
\begin{align}
\theta_k  = \arccos(\boldsymbol{r}\!_{\text{d}}^\mathrm{T}\cdot \boldsymbol{d}_k), \
\varphi_k  =\arccos(\boldsymbol{r}\!_{\text{a}}^\mathrm{T}\cdot \boldsymbol{d}_k).
\end{align}
Accordingly, $e_k$, which also denotes the normalized radiation pattern of a half-wavelength dipole antenna (cf. Fig.~\ref{System}), and $\tilde{\boldsymbol{a}}_k$ in \eqref{steering_vector} are modeled by \cite{r22}
\begin{alignat}{1}
&e_k = \frac{\cos \left(\frac{\pi}{2} \cos \theta_{k}\right)}{\sin \theta_{k}}, \ \theta_{k}\in[0,\pi],\ \text{and} \nonumber \\
&\tilde{\boldsymbol{a}}_k=(1, e^{j \frac{2 \pi}{\lambda} d \cdot \cos \varphi_{k}}, \cdots, e^{j(N-1) \frac{2 \pi}{\lambda} d \cdot \cos \varphi_{k}} )^{\mathrm{T}},
\label{phase_effect}
\end{alignat}
respectively, where $\lambda$ is the carrier wavelength and $d$ is the separation of adjacent dipoles in the RAA. To mitigate mutual coupling among antenna elements, we adopt $d=\lambda/2$ in this paper \cite{r22}. Finally, $\alpha=\sqrt{2}$ is a normalization factor introduced to ensure that a dipole radiates the same total power as an isotropic antenna.
\vspace{-0.5cm}
\subsection{Signal Model}
Now, define $\boldsymbol{w}_{k}\in\mathbb{C}^{N\times1}$ as the
beamforming vector for transmitting signal $s_{k}\in\mathbb{C}$ to
user $k$ with $\mathbb{E}[\left|s_{k}\right|^{2}]=1$, for all $k\in\mathcal{K}$. 
The resulting received signal of user $k$ is given by 
\begin{align}
y_{k}=\boldsymbol{h}_{k}^{\mathrm{H}}\sum\nolimits_{k=1}^{K}\boldsymbol{w}_{k}\cdot s_{k}+n_{k},
\end{align}
where $n_{k}\in\mathbb{C}$ is the additive white Gaussian noise (AWGN) at user $k$, following
$n_{k}\sim\mathcal{CN}(0,\sigma_{k}^{2})$. Consequently, the spectral efficiency of user $k$ in nat/s/Hz is given by
\begin{align}
R_{k}=\ln(1+\gamma_k),
\label{rate}
\end{align}
where $\gamma_k\in\mathbb{R}_+$ denotes the received signal-to-interference-plus-noise ratio (SINR) of user $k$, given by
\begin{align}
\gamma_k = \frac{|\boldsymbol{h}_{k}^{\mathrm{H}}  \boldsymbol{w}_{k}|^{2}}{\sum_{m=1, m \neq k}^{K}|\boldsymbol{h}_{k}^{\mathrm{H}}  \boldsymbol{w}_{m}|^{2}+\sigma^{2}_k}.
\label{SINR}
\end{align}
\vspace{-0.6cm}
\subsection{Problem Formulation}
Based on \eqref{rate} and \eqref{SINR}, the spectral efficiency $R_{k}$ depends on both the beamforming vectors $\{\boldsymbol{w}_1,\boldsymbol{w}_2,...,\boldsymbol{w}_K\}$ and the channel $\boldsymbol{h}_k$. The latter is further impacted by the RAA orientation, i.e., $\{\boldsymbol{r}\!_{\text{a}},\boldsymbol{r}\!_{\text{d}}\}$, which has to be judiciously adjusted to maximize the system performance while maintaining fairness in resource allocation across the $K$ users. To this end, we jointly design the RAA orientation and transmit beamforming vectors to maximize the sum-rate of all users subject to per-user QoS requirements. We define $\boldsymbol{W}\triangleq[\boldsymbol{w}_1,\boldsymbol{w}_2,...,\boldsymbol{w}_K]\in\mathbb{C}^{N\times K}$ and formulate the problem as
\begin{equation}
\begin{aligned}
& \underset{\boldsymbol{W},\boldsymbol{r}\!_{\text{a}},\boldsymbol{r}\!_{\text{d}}}{\text{max}} 
& & \sum\nolimits_{k=1}^{K} R_k\\
& \ \ \text{s.t.}
& & \text{C1:} \ \sum\nolimits_{k=1}^{K}\left\|\boldsymbol{w}_{k}\right\|^{2}\leq P, \\
& & & \text{C2:} \ \left\| \boldsymbol{r}\!_{\text{a}} \right\|=1, \ \text{C3:} \ \left\| \boldsymbol{r}\!_{\text{d}} \right\|=1, \\
& & & \text{C4:} \ \boldsymbol{r}\!_{\text{a}}^\mathrm{T}\cdot \boldsymbol{r}\!_{\text{d}}=0, \ \text{C5:} \ R_k \geq \overline{R}_k, \ k\in\mathcal{K}.
\end{aligned}
\label{P1}
\tag{P1}
\end{equation}
In problem \eqref{P1}, constraint C1 limits the maximal transmit power for communication to $P$. Constraints C2--C4 ensure that the direction vectors $\boldsymbol{r}\!_{\text{a}}$ and $\boldsymbol{r}\!_{\text{d}}$ have unit length and are orthogonal. Constraint C5 guarantees a minimum achievable rate of $\overline{R}_k\in\mathbb{R}_+$ for each user $k$, which constitutes the QoS requirement. C5 is necessary in practical deployment to prevent exceedingly low or even zero achievable rates for users with poor channels, striking a balance between performance and fairness. Since $R_k$ increases monotonically with the SINR $\gamma_k$ in \eqref{rate}, constraint C5 can be reformulated as
\begin{align}
\text{C6}:\ \gamma_k \geq \overline{\gamma}_k, \ k\in\mathcal{K},
\end{align}
where $\overline{\gamma}_k\!=\!e^{\overline{R}_k}\!-\!1$ is the minimum required SINR of user $k$. 

Problem \eqref{P1} is highly nonconvex because of the nonconcave objective function and nonconvex constraints C2--C4, C6, making it intractable to find a global optimal solution. Moreover, variables $\{\boldsymbol{W},\boldsymbol{r}\!_{\text{a}},\boldsymbol{r}\!_{\text{d}}\}$ are tightly coupled in the objective function and in constraints C4, C6. If the BCD method is applied to decompose \eqref{P1} into subproblems and solve them iteratively in a straightforward way, the attained solutions are usually inefficient. This is because the block variables $\{\boldsymbol{W}\}$, $\{\boldsymbol{r}\!_{\text{a}}\}$, $\{\boldsymbol{r}\!_{\text{d}}\}$ are mutually dependent and non-separable in C4 and C6 of \eqref{P1}. By projecting in brute force the overall feasible region of \eqref{P1} onto a block of variables in a given iteration, the BCD method is confined to solutions within a significantly reduced search space. To address this challenging problem, we propose an optimization framework based on the PDD method and a DL framework utilizing GNNs in Sections~III and~IV, respectively. The proposed optimization framework first decouples the variables in \eqref{P1} before applying BCD, offering more principled handling of constraints C4 and C6 and theoretical convergence guaranties. The proposed DL framework employs a two-module GNN-based architecture together with a two-stage training strategy to handle the coupled variables, approaching the performance of the optimization framework with substantially reduced computation time.

\emph{Remark~1}: In the system model and problem formulation, we assume that the UAV has perfect knowledge of the users’ positions. However, acquiring the precise user positions can be challenging in practice due to fluctuations of the UAV or mechanical vibrations of the RAA \cite{r23,r24}. We will investigate such imperfections and evaluate their impact on the achievable performance of the proposed solutions via simulations in Section~V. Interestingly, our results reveal that the DL framework exhibits significant robustness against user position errors. 

\emph{Remark~2}: Although we assume in this paper the UAV to be hovering, the considered system model and the proposed optimization and DL frameworks can also be applied to moving UAVs by updating the RAA orientation and beamforming based on the varying UAV positions, or by additionally optimizing the UAV trajectory. The latter could offer even higher sum-rate performance under the same QoS requirements, whereas, the resulting joint design problem becomes significantly more difficult to solve due to strong coupling between the UAV trajectory, RAA orientation, and beamforming. Due to the limited space, this is left as an interesting topic for future research.
\vspace{-0.35cm}
\section{Optimization Framework for Problem \eqref{P1}}
In this section, we propose an optimization framework based on the PDD method to address the nonconvex design problem \eqref{P1}. The basic idea behind the PDD method is to decouple the variables in the constraints by introducing auxiliary variables and constructing an augmented objective function involving penalty terms and dual variables. As such, the resulting problem can be decomposed and solved iteratively via a double-loop algorithm, where the inner loop solves the subproblems in a BCD manner, while the outer loop updates the dual variables and/or penalty parameters to attain a feasible and efficient solution to \eqref{P1} \cite{r27}.
\vspace{-0.45cm}
\subsection{Problem Reformulation and Decoupling of Variables}
To facilitate the PDD method, we define auxiliary variables
$x_{k,m}\!=\!\boldsymbol{h}_{k}^\mathrm{H}  \boldsymbol{w}_{m}\!\in\!\mathbb{C}, \forall k,m \!\in\! \mathcal{K}$, and rewrite the achievable rate of user $k$ in terms of $x_{k,m}$ as $\tilde{R}_k(\boldsymbol{x}_k)\!=\!\ln(1\!+\!\tilde{\gamma}_k(\boldsymbol{x}_k))$ with
\begin{align}
\tilde{\gamma}_k(\boldsymbol{x}_k)=\ln\Big(1+\frac{|x_{k,k}|^2}{\sum\nolimits_{m=1, m \neq k}^{K}|x_{k,m}|^2+\sigma^{2}_k}\Big),
\end{align}
where $\boldsymbol{x}_k\triangleq[x_{k,1},x_{k,2},...,x_{k,K}]^\mathrm{T}\in\mathbb{C}^{K\times 1}$ and $\tilde{\gamma}_k(\boldsymbol{x}_k)$ represents the received SINR of user $k$ in terms of $\boldsymbol{x}_{k}$. Subsequently, using $\boldsymbol{x}_k$, constraint C6 is rewritten as
\begin{align}
\text{C7}:\ \overline{\gamma}_k&(\sum\nolimits_{m=1, m \neq k}^{K}|x_{k,m}|^2+\sigma^{2}_k)-|x_{k,k}|^2 \leq 0, \forall k. %\in\mathcal{K}.
\end{align}
By defining $\boldsymbol{X}\!\triangleq\![\boldsymbol{x}_1,\boldsymbol{x}_2,...,\boldsymbol{x}_K]^{\mathrm{T}}\!\in\!\mathbb{C}^{K\times K}$ and $\boldsymbol{H}\!\triangleq\![\boldsymbol{h}_1,\boldsymbol{h}_2,...,\boldsymbol{h}_K]\!\in\!\mathbb{C}^{N\times K}$, problem \eqref{P1} can be equivalently reformulated as
\begin{equation}
\begin{aligned}
& \underset{\boldsymbol{X},\boldsymbol{W},\boldsymbol{r}\!_{\text{a}},\boldsymbol{r}\!_{\text{d}}}{\text{max}} 
& & \sum\nolimits_{k=1}^{K} \tilde{R}_k(\boldsymbol{x}_k)\\
& \ \ \ \text{ s.t.}
& & \text{C1, C2, C3, C4, C7,} \\
& & & \text{C8:} \ \boldsymbol{X}=\boldsymbol{H}^\mathrm{H}\boldsymbol{W}.
\end{aligned}
\label{P1'}
\tag{$\widetilde{\text{P1}}$}
\end{equation}

Problem \eqref{P1'} remains nonconvex because of the nonconcave objective function and nonconvex constraints C2--C4 and C7. However, variables $\{ \boldsymbol{X},\boldsymbol{W},\boldsymbol{r}\!_{\text{a}},\boldsymbol{r}\!_{\text{d}}\}$ are only coupled in equality constraints C4 and C8, and the remaining problem (excluding C4 and C8) is \emph{block-separable} w.r.t. $\{\boldsymbol{X}\}$, $\{\boldsymbol{W}\}$, $\{\boldsymbol{r}\!_{\text{a}}\}$, and $\{\boldsymbol{r}\!_{\text{d}}\}$. To exploit these properties via PDD, we further define the augmented Lagrangian (AL) problem of \eqref{P1'} as \cite{r28}
\begin{alignat}{1}
%\begin{aligned}
%&
&\underset{\boldsymbol{X},\boldsymbol{W},\boldsymbol{r}\!_{\text{a}},\boldsymbol{r}\!_{\text{d}}}{\text{max}}
\mathcal{L}_\text{PDD}=\sum\nolimits_{k=1}^{K} \tilde{R}_k(\boldsymbol{x}_k) \nonumber \\
&
-\frac{1}{2\rho_1}\| \boldsymbol{X}-\boldsymbol{H}^\mathrm{H}\boldsymbol{W}+\rho_1 \boldsymbol{Z}_1 \|_{\text{F}}^2-\frac{1}{2\rho_2}|\boldsymbol{r}\!_{\text{a}}^\mathrm{T}\boldsymbol{r}\!_{\text{d}}+\rho_2z_2|^2 \label{P2}
\tag{P2}
\\
& \text{ s.t. \; }
\text{ C1, C2, C3, C7,} \nonumber
\end{alignat}
where C4 and C8 in \eqref{P1'} are absorbed into the objective function of \eqref{P2} using dual variables and penalty parameters. Particularly, $\boldsymbol{Z}_1\!\in\!\mathbb{C}^{K\times K}$ ($\rho_1\!\in\!\mathbb{R}_+$) and $z_2\!\in\!\mathbb{R}$ ($\rho_2\!\in\!\mathbb{R}_+$) denote the dual variables (penalty parameters) for C8 and C4, respectively. As expected, variables $\{ \boldsymbol{X},\boldsymbol{W},\boldsymbol{r}\!_{\text{a}},\boldsymbol{r}\!_{\text{d}}\}$ in \eqref{P2} are no longer coupled in the constraints, providing a favorable decomposable structure. With this, we are ready to apply the PDD method to solve problem \eqref{P2} in two loops: In the inner iteration, both the penalty parameters and the dual variables are fixed and \eqref{P2} can be efficiently solved in a BCD manner for the block variables $\{\boldsymbol{W}\}$, $\{\boldsymbol{X}\}$, $\{\boldsymbol{r}\!_{\text{a}}\}$, and $\{\boldsymbol{r}\!_{\text{d}}\}$; in the outer iteration, we update the dual variables and the penalty parameters to gradually enforce equality constraints C4 and C8 and improve the quality of the overall solution.
\vspace{-0.2cm}
\subsection{Inner Iteration}
In this part, we solve \eqref{P2} for given penalty parameters and dual variables. Exploiting its block-separable structure, we decompose \eqref{P2} into four subproblems, each of which optimizes one block variable in $\{\boldsymbol{W}\}$, $\{\boldsymbol{X}\}$, $\{\boldsymbol{r}\!_{\text{a}}\}$, and $\{\boldsymbol{r}\!_{\text{d}}\}$. Closed-form solutions for these subproblems are derived.
\subsubsection{Optimization of $\boldsymbol{W}$} For fixed $\{\boldsymbol{X},\boldsymbol{r}\!_{\text{a}},\boldsymbol{r}\!_{\text{d}}\}$, the subproblem of optimizing $\boldsymbol{W}$ becomes
\begin{equation}
\begin{aligned}
& \underset{\boldsymbol{W}}{\text{min}}
& & \| \boldsymbol{X}-\boldsymbol{H}^\mathrm{H}\boldsymbol{W}+\rho_1 \boldsymbol{Z}_1 \|_{\text{F}}^2 \\
& \text{ s.t.}
& & \widetilde{\text{C1}}: \ \| \boldsymbol{W}\|_{\text{F}}^2\leq P, \\
\end{aligned}
\label{P3}
\tag{P3}
\end{equation}
where constraint $\widetilde{\text{C1}}$ is an equivalent reformulation of C1. Problem \eqref{P3} is convex and satisfies the Slater's condition, for which strong duality holds. Thus, \eqref{P3} can be tackled by solving its Lagrangian dual problem \cite{r29}. The optimal $\boldsymbol{W}$ is
\begin{align}
\boldsymbol{W}^\star=\boldsymbol{W}(\lambda_1^\star)=(\boldsymbol{H}\boldsymbol{H}^H+\lambda_1^\star\boldsymbol{I}\!_N)^{-1}\boldsymbol{H}(\boldsymbol{X}+\rho_1 \boldsymbol{Z}_1),
\end{align}
where $\lambda_1\geq0$ denotes the Lagrangian multiplier of $\widetilde{\text{C1}}$ and $\lambda_1^\star$ denotes its optimal value. If $\| \boldsymbol{W}(0) \|_{\text{F}}^2\leq P$, $\lambda_1^\star=0$; otherwise, $\lambda_1^\star$ can be found by solving the equation $\| \boldsymbol{W}(\lambda_1) \|_{\text{F}}^2=P$ via a bisection search \cite{r29}.
\subsubsection{Optimization of $\boldsymbol{X}$} For fixed $\{\boldsymbol{W},\boldsymbol{r}\!_{\text{a}},\boldsymbol{r}\!_{\text{d}}\}$, the subproblem of optimizing $\boldsymbol{X}$ is given by
\begin{equation}
\begin{aligned}
& \underset{\boldsymbol{X}}{\text{max}}
& & \sum\nolimits_{k=1}^{K} \tilde{R}_k(\boldsymbol{x}_k)
-\frac{1}{2\rho_1}\| \boldsymbol{X}-\boldsymbol{H}^\mathrm{H}\boldsymbol{W}+\rho_1 \boldsymbol{Z}_1 \|_{\text{F}}^2 \\
& \text{ s.t.}
& & \text{C7.}
\end{aligned}
\label{P4}
\tag{P4}
\end{equation}
\eqref{P4} is nonconvex due to the nonconcave component $\sum\nolimits_{k=1}^{K} \tilde{R}_k(\boldsymbol{x}_k)$ of the objective function and nonconvex constraint C7. By deriving a concave surrogate of $\tilde{R}_k(\boldsymbol{x}_k)$, we solve \eqref{P4} iteratively using the majorization-minimization (MM) algorithm \cite{r31}. To this end, we define $i$ as the iteration index for solving \eqref{P4} and exploit the following theorem.
\begin{theorem}
Given any $\boldsymbol{x}_k^{(i)}$, $\tilde{R}_k(\boldsymbol{x}_k)$ is lower bounded by 
\begin{align}
&\tilde{R}_k(\boldsymbol{x}_k)\geq f(\boldsymbol{x}_k,\boldsymbol{x}_k^{(i)}) \notag \\
&=\ln(1+\tilde{\gamma}_k(\boldsymbol{x}_k^{(i)}))-\tilde{\gamma}_k(\boldsymbol{x}_k^{(i)}) \label{LB} \\
&+2\Re\left\{ \frac{x_{k,k}^{(i)*}x_{k,k}}{\zeta_k(\boldsymbol{x}_k^{(i)})} \right\}-\left(\frac{1}{\zeta_k(\boldsymbol{x}_k^{(i)})}-\frac{1}{\eta_k(\boldsymbol{x}_k^{(i)})}\right)\eta_k(\boldsymbol{x}_k), \notag
\end{align}
where $\zeta_k(\boldsymbol{x}_k)=\sum_{m=1, m \neq k}^{K}|x_{k,m}|^2+\sigma^{2}_k$, $\eta_k(\boldsymbol{x}_k)=\zeta_k(\boldsymbol{x}_k)+|x_{k,k}|^2$, and equality holds in \eqref{LB} when $\boldsymbol{x}_k=\boldsymbol{x}_k^{(i)}$.
\end{theorem}
The proof of Theorem 1 is omitted due to the limited space and can be obtained following similar steps as in \cite{r30}. Note that in Theorem 1, $f(\boldsymbol{x}_k,\boldsymbol{x}_k^{(i)})$ is a concave function that has the same value and first-order derivative as $\tilde{R}_k(\boldsymbol{x}_k)$ at $\boldsymbol{x}_k^{(i)}$. Thus, $f(\boldsymbol{x}_k,\boldsymbol{x}_k^{(i)})$ provides a concave surrogate function that tightly lower-bounds $\tilde{R}_k(\boldsymbol{x}_k)$. Given an initial solution, denoted by $\boldsymbol{X}^{(i)}$, we can then apply the MM algorithm to search for an improved solution, denoted as $\boldsymbol{X}^{(i+1)}$, by solving
\begin{equation}
\begin{aligned}
& \underset{\boldsymbol{X}}{\text{min}}
& & -\sum\nolimits_{k=1}^{K} f(\boldsymbol{x}_k,\boldsymbol{x}_k^{(i)}) \\
&&&+\frac{1}{2\rho_1}\| \boldsymbol{X}-\boldsymbol{H}^\mathrm{H}\boldsymbol{W}+\rho_1 \boldsymbol{Z}_1 \|_{\text{F}}^2 \\
& \text{ s.t.}
& & \text{C7.}
\end{aligned}
\label{P5}
\tag{P5}
\end{equation}
Compared with $\boldsymbol{X}^{(i)}$, $\boldsymbol{X}^{(i+1)}$ also provides an improved solution for \eqref{P4}. As such, by iteratively replacing the intractable part $\sum\nolimits_{k=1}^{K} \tilde{R}_k(\boldsymbol{x}_k)$ in the objective of \eqref{P4} with a tractable surrogate $f(\boldsymbol{x}_k,\boldsymbol{x}_k^{(i)})$ and solving the resulting problem \eqref{P5}, the MM algorithm will converge to a locally optimal solution of \eqref{P4} \cite{r31}. Note that interestingly, although \eqref{P5} is still nonconvex due to the nonconvex constraint C7, it is also tractable. To see this, we observe that \eqref{P5} can be decomposed into $K$ independent subproblems, each optimizing a row vector of matrix $\boldsymbol{X}$, i.e., $\boldsymbol{x}_k$, $k\in\mathcal{K}$, as follows,
%\begin{equation}
\begin{alignat}{2}
& \underset{\boldsymbol{x}_k}{\text{min}} & & - f(\boldsymbol{x}_k,\boldsymbol{x}_k^{(i)}) \nonumber \\
&&&+\frac{1}{2\rho_1}\sum\nolimits_{m=1}^{K}|x_{k,m}-\boldsymbol{h}_k^{\mathrm{H}}\boldsymbol{w}_m+\rho_1\boldsymbol{Z}_{1(k,m)}|^2 \label{P6}
\tag{P6} \\
& \text{ s.t.}
& & \text{C8:} \ \overline{\gamma}_k\Big(\sum_{m=1, m \neq k}^{K}|x_{k,m}|^2+\sigma^{2}_k\Big)-|x_{k,k}|^2 \leq 0. \nonumber
\end{alignat}
%\end{equation}
Problem \eqref{P6} is a quadratic program subject to a single quadratic constraint, for which strong duality also holds according to \cite[Appendix~B.1]{r29}. As a result, the KKT conditions necessarily hold for the optimal solution. Define $\boldsymbol{x}_k^\star$ and $\lambda_{2,k}^\star$ as the optimal $\boldsymbol{x}_k$ and the optimal Lagrangian multiplier w.r.t. C8, respectively. Now, we obatin the following theorem:
\begin{theorem}
By solving the KKT conditions of \eqref{P6}, we have
\begin{align}
x_{k,m}^\star=\frac{\boldsymbol{h}_k^\mathrm{H}\boldsymbol{w}_m-\rho_1\boldsymbol{Z}_{1(k,m)}}{2\rho_1(\delta(\boldsymbol{x}_k^{(i)})+\lambda_{2,k}^\star\overline{\gamma}_k)},\ \text{for }m\neq k,
\label{xkm}
\end{align}
where $\delta(\boldsymbol{x}_k)=\frac{1}{\zeta(\boldsymbol{x}_k)}-\frac{1}{\eta(\boldsymbol{x}_k)}-\frac{1}{2\rho_1}$. If $\frac{1}{2\rho_1}(\rho_1\boldsymbol{Z}_{1(k,k)}-\boldsymbol{h}_k^\mathrm{H}\boldsymbol{w}_k)-x_{k,k}^{(i)}/\zeta_k(\boldsymbol{x}_k^{(i)})=0$, then $\lambda_{2,k}^\star=\delta(\boldsymbol{x}_k^{(i)})$ and
\begin{align}
x_{k,k}^\star=\sqrt{\overline{\gamma}_k(\sum\nolimits_{m=1, m \neq k}^{K}|x_{k,m}^\star|^2+\sigma^{2}_k)}\cdot e^{j\iota},
\label{xkk1}
\end{align}
where $\iota\in\mathbb{R}$ is an arbitrary value; otherwise,
\begin{align}
x_{k,k}^\star=\frac{\frac{1}{2\rho_1}(\rho_1\boldsymbol{Z}_{1(k,k)}-\boldsymbol{h}_k^\mathrm{H}\boldsymbol{w}_k)-x_{k,k}^{(i)}/\zeta_k(\boldsymbol{x}_k^{(i)})}{\delta(\boldsymbol{x}_k^{(i)})-\lambda_{2,k}^\star},
\label{xkk2}
\end{align}
where $\lambda_{2,k}^\star$ in \eqref{xkk2} can be found by solving the equation 
\begin{align}
\lambda_{2,k}^\star\overline{\gamma}_k\big(\sum\nolimits_{m=1, m \neq k}^{K}|x_{k,m}^\star|^2+\sigma^{2}_k\big)-\lambda_{2,k}^\star|x_{k,k}^\star|^2=0,
\end{align}
via a bisection search.
\end{theorem}
\begin{proof}
    Please refer to the Appendix.
\end{proof}
\subsubsection{Optimization of $\boldsymbol{r}\!_{\text{a}}$} For fixed $\{\boldsymbol{W},\boldsymbol{X},\boldsymbol{r}\!_{\text{d}}\}$, the subproblem of optimizing $\boldsymbol{r}\!_{\text{a}}$ is given by
\begin{alignat}{1}
\underset{\boldsymbol{r}\!_{\text{a}}}{\text{min}}
 & \quad g(\boldsymbol{r}\!_{\text{a}}) \!=\! \tfrac{1}{2\rho_1} \| \boldsymbol{X}\!-\!\boldsymbol{H}^\mathrm{H}\boldsymbol{W}\!+\!\rho_1 \boldsymbol{Z}_1 \|_{\text{F}}^2  \!+\!\tfrac{1}{2\rho_2} |\boldsymbol{r}\!_{\text{a}}^\mathrm{T}\boldsymbol{r}\!_{\text{d}}\!+\!\rho_2z_2|^2 \nonumber \\
\text{ s.t.}
 & \text{\quad  C2.}
%\end{aligned}
\label{P7}
\tag{P7}
\end{alignat}
\eqref{P7} is nonconvex due to the nonconvex objective $g(\boldsymbol{r}\!_{\text{a}})$ and the constant-norm constraint C2. However, since C2 indicates that $\boldsymbol{r}\!_{\text{a}}$ lies on a 3D unit sphere manifold $\mathcal{M}\triangleq\{ \boldsymbol{s}|\boldsymbol{s}\in\mathbb{R}^{3\times 1},\| \boldsymbol{s}\|^2=1 \}$ and $g(\boldsymbol{r}\!_{\text{a}})$ is differentiable w.r.t. $\boldsymbol{r}\!_{\text{a}}$, \eqref{P7} is a manifold optimization problem defined on the smooth manifold $\mathcal{M}$ \cite{r32} and can be efficiently tackled by manifold optimization methods. Thus, we propose to adopt the Riemannian conjugate gradient (RCG) method to solve \eqref{P7}. By extending the classical conjugate gradient (CG) method to Riemannian manifolds, the RCG method aims to iteratively perform efficient descent steps along the Riemannian gradient by directly exploiting the geometric structure of the manifold. 

To this end, let us first derive the Euclidean gradient
\begin{align}
\nabla_{\boldsymbol{r}\!_{\text{a}}}g(\boldsymbol{r}\!_{\text{a}})\!=\!2\begin{bmatrix} \Re\{\text{tr} \{ (\nabla_{\boldsymbol{H}}g_1(\boldsymbol{r}\!_{\text{a}}))^{\mathrm{H}}\frac{\partial\boldsymbol{H}}{\partial r\!_{a,1}} \}\}\\ \Re\{\text{tr} \{ (\nabla_{\boldsymbol{H}}g_1(\boldsymbol{r}\!_{\text{a}}))^{\mathrm{H}}\frac{\partial\boldsymbol{H}}{\partial r\!_{a,2}}  \}\}\\ \Re\{\text{tr} \{ (\nabla_{\boldsymbol{H}}g_1(\boldsymbol{r}\!_{\text{a}}))^{\mathrm{H}}\frac{\partial\boldsymbol{H}}{\partial r\!_{a,3}}  \}\} \end{bmatrix}\!\!+\!\frac{(\boldsymbol{r}\!_{\text{a}}^\mathrm{T}\boldsymbol{r}\!_{\text{d}}+\rho_2z_2)\boldsymbol{r}\!_{\text{d}}}{\rho_2},
\label{EuG}
\end{align}
where $g_1(\boldsymbol{r}\!_{\text{a}})=\frac{\| \boldsymbol{X}-\boldsymbol{H}^\mathrm{H}\boldsymbol{W}+\rho_1 \boldsymbol{Z}_1 \|_{\text{F}}^2}{2\rho_1}$ and $\nabla_{\boldsymbol{H}}g_1(\boldsymbol{r}\!_{\text{a}})=-\frac{1}{2\rho_1}\boldsymbol{W}(\boldsymbol{X}-\boldsymbol{H}^\mathrm{H}\boldsymbol{W}+\rho_1 \boldsymbol{Z}_1)^{\mathrm{H}}$. $\frac{\partial\boldsymbol{H}}{\partial r\!_{a,i}}\in\mathbb{C}^{N\times K}$ represents the partial derivative of $\boldsymbol{H}$ w.r.t. the $i$-th element of $\boldsymbol{r}\!_{\text{a}}$ with
\begin{align}
{\frac{\partial\boldsymbol{H}}{\partial r\!_{a,i}}}_{(n,k)}=(n-1)\alpha\sqrt{\beta_k}e_kj\pi e^{j(n-1)\pi \boldsymbol{r}\!_{\text{a}}^\mathrm{T}\cdot \boldsymbol{d}_k}d_{k,i}.
\end{align}
However, since $\boldsymbol{r}\!_{\text{a}}$ is constrained to lie on $\mathcal{M}$ rather than in the Euclidean space, $\nabla_{\boldsymbol{r}\!_{\text{a}}}g(\boldsymbol{r}\!_{\text{a}})$ does not lead to a descent direction. In contrast, the Riemannian gradient $\text{grad}_{\boldsymbol{r}\!_{\text{a}}}g(\boldsymbol{r}\!_{\text{a}})$ defines the steepest descent direction on $\mathcal{M}$ by further projecting $\nabla_{\boldsymbol{r}\!_{\text{a}}}g(\boldsymbol{r}\!_{\text{a}})$ onto the tangent space of $\mathcal{M}$ with \cite{r32}
\begin{align}
\text{grad}_{\boldsymbol{r}\!_{\text{a}}}g(\boldsymbol{r}\!_{\text{a}})=(\boldsymbol{I}_3-\boldsymbol{r}\!_{\text{a}}\boldsymbol{r}\!_{\text{a}}^\mathrm{T})\nabla_{\boldsymbol{r}\!_{\text{a}}}g(\boldsymbol{r}\!_{\text{a}}).
\label{ReG}
\end{align}

Then, the RCG method performs descent along the Riemannian gradient in \eqref{ReG} and retracts the resulting point from the tangent space back onto the unit sphere manifold $\mathcal{M}$. Let $\boldsymbol{r}\!_{\text{a}}^{(l-1)}$ denote the point obtained in the $(l-1)$-th iteration of the RCG method. Then, the update rule for $\boldsymbol{r}\!_{\text{a}}$ is
\begin{align}
\boldsymbol{r}\!_{\text{a}}^{(l)}=\frac{\boldsymbol{r}\!_{\text{a}}^{(l-1)}-c_0^{(l)} \text{grad}_{\boldsymbol{r}\!_{\text{a}}}g(\boldsymbol{r}\!_{\text{a}}^{(l-1)})}{\| \boldsymbol{r}\!_{\text{a}}^{(l-1)}-c_0^{(l)} \text{grad}_{\boldsymbol{r}\!_{\text{a}}}g(\boldsymbol{r}\!_{\text{a}}^{(l-1)}) \|},
\label{UpdatingRCG}
\end{align}
where $c_0^{(l)}$ denotes the step size at the $l$-th iteration and is determined by a line search using, e.g., the Armijo backtracking method \cite{r33}. The overall algorithm for solving \eqref{P7} is summarized in Algorithm 1. As shown in \cite{r32}, the RCG method is guaranteed to converge to a critical point of \eqref{P7}.
\renewcommand{\algorithmicrequire}{\textbf{Input:}}
\begin{algorithm}[!t]
\footnotesize
\caption{The RCG Method for Solving Problem \eqref{P7}}
\begin{algorithmic}[1]
\REQUIRE $\boldsymbol{W}$, $\boldsymbol{X}$, $\boldsymbol{r}\!_{\text{d}}$, initial $\boldsymbol{r}\!_{\text{a}}^{(0)}$ and $l=1$.
\REPEAT
\STATE Calculate the Euclidean gradient $\nabla_{\boldsymbol{r}\!_{\text{a}}}g(\boldsymbol{r}\!_{\text{a}}^{(l-1)})$  with \eqref{EuG}.
\STATE Calculate the Riemannian gradient $\text{grad}_{\boldsymbol{r}\!_{\text{a}}}g(\boldsymbol{r}\!_{\text{a}}^{(l-1)})$ with \eqref{ReG}.
\STATE Choose a proper step size $\mu^{(l)}$ with line search.
\STATE Update $\boldsymbol{r}\!_{\text{a}}^{(l)}$ with \eqref{UpdatingRCG}, and $l=l+1$.
\UNTIL{$\| \text{grad}_{\boldsymbol{r}\!_{\text{a}}}g(\boldsymbol{r}\!_{\text{a}}^{(l)}) \|$} is sufficiently small.
\end{algorithmic}
\end{algorithm}
\subsubsection{Optimization of $\boldsymbol{r}\!_{\text{d}}$} For fixed $\{\boldsymbol{W},\boldsymbol{X},\boldsymbol{r}\!_{\text{a}}\}$, the subproblem of optimizing $\boldsymbol{r}\!_{\text{d}}$ is given by
\begin{equation}
\begin{aligned}
& \underset{\boldsymbol{r}\!_{\text{d}}}{\text{min}}
& & \tfrac{1}{2\rho_1} \| \boldsymbol{X} \!-\! \boldsymbol{H}^\mathrm{H}\boldsymbol{W} \!+\! \rho_1 \boldsymbol{Z}_1 \|_{\text{F}}^2  +\tfrac{1}{2\rho_2} |\boldsymbol{r}\!_{\text{a}}^\mathrm{T}\boldsymbol{r}\!_{\text{d}}+\rho_2z_2|^2 \\
& \text{ s.t.}
& & \text{C3.}
\end{aligned}
\label{P8}
\tag{P8}
\end{equation}
Since $\boldsymbol{r}\!_{\text{d}}$ also lies on a 3D unit sphere manifold $\mathcal{M}$, problem \eqref{P8} can be solved in a similar manner to \eqref{P7} using Algorithm~1. Due to the limited space, the details are omitted here.
\vspace{-0.6cm}
\subsection{Outer Iteration and Initial Solution}
As discussed in Section~III--B, in the inner iteration, variables $\{\boldsymbol{W}$, $\boldsymbol{X}$, $\boldsymbol{r}\!_{\text{a}}$, $\boldsymbol{r}\!_{\text{d}}\}$ are updated while the dual variables $\{\boldsymbol{Z}_1$, $z_2\}$ and penalty parameters $\{\rho_1$, $\rho_2\}$ are kept fixed. Upon convergence in the inner iteration, $\{\boldsymbol{Z}_1$, $z_2\}$ or $\{\rho_1$, $\rho_2\}$ are updated in the outer iteration to gradually enforce equality constraints C4 and C8, which are handled through the augmented Lagrangian penalties in \eqref{P2}.
\subsubsection{Update $\boldsymbol{Z}_1$, $z_2$, $\rho_1$ and $\rho_2$} Let $t_{\text{in}}$ and $t_{\text{out}}$ denote indices of the inner and outer iterations, respectively. After the inner loop has converged, in the $t_\text{out}$-th outer iteration, the constraint violations are first determined by calculating $\| \boldsymbol{X}-\boldsymbol{H}^\mathrm{H}\boldsymbol{W} \|_{\text{F}}^2$ and $|\boldsymbol{r}\!_{\text{a}}^\mathrm{T}\boldsymbol{r}\!_{\text{d}}|^2$. If $\| \boldsymbol{X}-\boldsymbol{H}^\mathrm{H}\boldsymbol{W} \|_{\text{F}}^2\leq \epsilon_1$ (or $|\boldsymbol{r}\!_{\text{a}}^\mathrm{T}\boldsymbol{r}\!_{\text{d}}|^2\leq \epsilon_2$) holds for given tolerance $\epsilon_1>0$ (or $\epsilon_2>0$), the current solution of $\{\boldsymbol{W}$, $\boldsymbol{X}$, $\boldsymbol{r}\!_{\text{a}}$, $\boldsymbol{r}\!_{\text{d}}\}$ is considered approximately satisfied for C8 (or C4), and only the dual variable $\boldsymbol{Z}_1$ (or $z_2$) is updated by accumulating the residual violation of C8 (or C4) as
\begin{align}
&\boldsymbol{Z}_1^{(t_{\text{out}}+1)}=\boldsymbol{Z}_1^{(t_{\text{out}})}+(\boldsymbol{X}-\boldsymbol{H}^\mathrm{H}\boldsymbol{W})/\rho_1^{(t_{\text{out}})} \label{Z1update} \\
(&\text{or } z_2^{(t_{\text{out}}+1)}=z_2^{(t_{\text{out}})}+(\boldsymbol{r}\!_{\text{a}}^\mathrm{T}\boldsymbol{r}\!_{\text{d}})/\rho_2^{(t_{\text{out}})} ),\label{z2update}
\end{align}
which will improve the quality of the overall solution in the next outer iteration. Otherwise, we decrease the penalty parameter $\rho_1$ (or $\rho_2$) as $\rho_1^{(t_{\text{out}}+1)}=c_1\rho_1^{(t_{\text{out}})}$ (or $\rho_2^{(t_{\text{out}}+1)}=c_2\rho_2^{(t_{\text{out}})}$), where $0<c_1<1$ (or $0<c_2<1$), to impose a stronger penalty on constraint C8 (or C4) and improve the feasibility of solution in the next outer iteration.
\renewcommand{\algorithmicrequire}{\textbf{Input:}}
\begin{algorithm}[!t]
\footnotesize
\caption{The PDD Algorithm for Solving Problem \eqref{P2}}
\begin{algorithmic}[1]
\REQUIRE Initialize $t_{\text{out}}=0$, $\boldsymbol{W}^{(t_{\text{out}})}$, $\boldsymbol{X}^{(t_{\text{out}})}$, $\boldsymbol{r}\!_{\text{d}}^{(t_{\text{out}})}$, $\boldsymbol{r}\!_{\text{a}}^{(t_{\text{out}})}$, $\boldsymbol{Z}_1^{(t_{\text{out}})}$, $z_2^{(t_{\text{out}})}$, $\rho_1^{(t_{\text{out}})}$, $\rho_2^{(t_{\text{out}})}$, set the tolerances $\epsilon_1$, $\epsilon_2$, $\epsilon_3$, $\epsilon_4$, the reduction factors $c_1$, $c_2$, the minimum penalty parameters $\rho_{1,\text{min}}$, $\rho_{2,\text{min}}$ and the maximum number of inner and outer iterations $T_{\text{in}}^{\text{max}}$ and $T_{\text{out}}^{\text{max}}$, respectively.
\REPEAT
\STATE Initialize $t_{\text{in}}=0$, $\boldsymbol{W}^{(t_{\text{out}},t_{\text{in}})}=\boldsymbol{W}^{(t_{\text{out}})}$, $\boldsymbol{X}^{(t_{\text{out}},t_{\text{in}})}=\boldsymbol{X}^{(t_{\text{out}})}$, $\boldsymbol{r}\!_{\text{a}}^{(t_{\text{out}},t_{\text{in}})}=\boldsymbol{r}\!_{\text{a}}^{(t_{\text{out}})}$, and $\boldsymbol{r}\!_{\text{d}}^{(t_{\text{out}},t_{\text{in}})}=\boldsymbol{r}\!_{\text{d}}^{(t_{\text{out}})}$.
\REPEAT
\STATE Update $\boldsymbol{W}^{(t_{\text{out}},t_{\text{in}}+1)}$ by solving problem \eqref{P3}.
\STATE Update $\boldsymbol{X}^{(t_{\text{out}},t_{\text{in}}+1)}$ by solving problem \eqref{P4}.
\STATE Update $\boldsymbol{r}\!_{\text{a}}^{(t_{\text{out}},t_{\text{in}}+1)}$ by solving problem \eqref{P7}.
\STATE Update $\boldsymbol{r}\!_{\text{d}}^{(t_{\text{out}},t_{\text{in}}+1)}$ by solving problem \eqref{P8}.
\STATE $t_{\text{in}}=t_{\text{in}}+1$.
\UNTIL $|\mathcal{L}_{\text{PDD}}^{(t_{\text{out}},t_{\text{in}})}-\mathcal{L}_{\text{PDD}}^{(t_{\text{out}},t_{\text{in}}-1)}|\leq \epsilon_3$ or $t_{\text{in}}>T_{\text{in}}^{\text{max}}$.
\STATE \textbf{If} $(\| \boldsymbol{X}-\boldsymbol{H}^\mathrm{H}\boldsymbol{W} \|_{\text{F}}^2)^{(t_{\text{out}})}<\epsilon_1$, \textbf{then} update $\boldsymbol{Z}_1^{(t_{\text{out}}+1)}$ with \eqref{Z1update}; \textbf{else if} $\rho_1^{(t_{\text{out}})}>\rho_{1,\text{min}}$, \textbf{then} $\rho_1^{(t_{\text{out}}+1)}=c_1\rho_1^{(t_{\text{out}})}$; \textbf{else} $\rho_1^{(t_{\text{out}}+1)}=\rho_{1,\text{min}}$.
\STATE \textbf{If} $(|\boldsymbol{r}\!_{\text{a}}^\mathrm{T}\boldsymbol{r}\!_{\text{d}}|^2)^{(t_{\text{out}})}<\epsilon_2$, \textbf{then} update $z_2^{(t_{\text{out}}+1)}$ with \eqref{z2update}; \textbf{else if} $\rho_2^{(t_{\text{out}})}>\rho_{2,\text{min}}$, \textbf{then} $\rho_2^{(t_{\text{out}}+1)}=c_2\rho_2^{(t_{\text{out}})}$; \textbf{else} $\rho_2^{(t_{\text{out}}+1)}=\rho_{2,\text{min}}$.
\STATE $t_{\text{out}}=t_{\text{out}}+1$.
\UNTIL $|\mathcal{L}_{\text{PDD}}^{(t_{\text{out}})}-\mathcal{L}_{\text{PDD}}^{(t_{\text{out}}-1)}|\leq \epsilon_4$ and $(\| \boldsymbol{X}-\boldsymbol{H}^\mathrm{H}\boldsymbol{W} \|_{\text{F}}^2)^{(t_{\text{out}})}\leq\epsilon_1$ and $(|\boldsymbol{r}\!_{\text{a}}^T\boldsymbol{r}\!_{\text{d}}|^2)^{(t_{\text{out}})}\leq\epsilon_2$, or $t_{\text{out}}>T_{\text{out}}^{\text{max}}$.
\end{algorithmic}
\end{algorithm}
\subsubsection{Initialization of $\boldsymbol{X}$, $\boldsymbol{W}$, $\boldsymbol{r}\!_{\text{a}}$, and $\boldsymbol{r}\!_{\text{d}}$}
An initial solution is needed to start the double-loop iteration, and an effective initial solution of \eqref{P2} can accelerate the convergence of the PDD algorithm and improve the resulting sum-rate performance. It has been shown in \cite{r11} that, by orientation design, an RAA composed of an infinite number of antennas can always ensure the pairwise orthogonality of the users' channels, eliminating inter-user channel correlation. Although this result is generally unattainable for an RAA composed of a finite number of antennas, it suggests a practical and effective initial design for RAA orientation by minimizing the inter-user channel correlation. Inspired by this, we initialize $\boldsymbol{r}\!_{\text{a}}$ heuristically as
\begin{equation}
\begin{aligned}
& \underset{\boldsymbol{r}\!_{\text{a}}}{\text{min}}
& & \sum\nolimits_{k=1}^{K}\sum\nolimits_{m=k+1}^{K}|\tilde{\boldsymbol{a}}_k^\mathrm{T}\tilde{\boldsymbol{a}}_m|^2 \\
& \text{ s.t.}
& & \text{C2,}
\end{aligned}
\label{P9}
\tag{P9}
\end{equation}
where $|\tilde{\boldsymbol{a}}_k^\mathrm{T}\tilde{\boldsymbol{a}}_m|^2$ represents the channel correlation between users $k$ and $m$. Problem \eqref{P9} can be addressed using the RCG method as in Algorithm~1. Meanwhile, an arbitrary unit vector orthogonal to the initial $\boldsymbol{r}\!_{\text{a}}$ provides an initial $\boldsymbol{r}\!_{\text{d}}$ satisfying constraint C4. Given the initial $\boldsymbol{r}\!_{\text{a}}$ and $\boldsymbol{r}\!_{\text{d}}$, finding an initial $\boldsymbol{W}$ that satisfies constraints C1 and C5 can be formulated as a second-order cone programming (SOCP) problem \cite{r34} and efficiently solved by the interior-point method. Finally, $\boldsymbol{X}$ is initialized as $ \boldsymbol{X}=\boldsymbol{H}^\mathrm{H}\boldsymbol{W}$ to satisfy constraint C8.
\vspace{-0.5cm}
\subsection{Overall Solution of \eqref{P2} and Computational Complexity}
The PDD algorithm terminates when the objective $\mathcal{L}_\text{PDD}$ of \eqref{P2} converges and the constraint violations become negligible within tolerances, i.e., $\| \boldsymbol{X}-\boldsymbol{H}^\mathrm{H}\boldsymbol{W} \|_{\text{F}}^2\leq \epsilon_1$ and $|\boldsymbol{r}\!_{\text{a}}^\mathrm{T}\boldsymbol{r}\!_{\text{d}}|^2\leq \epsilon_2$ are satisfied. The overall PDD algorithm for solving \eqref{P2} is summarized in Algorithm 2. As can be inferred from \cite{r25,r26,r27}, Algorithm~2 is guaranteed to converge to a KKT solution of \eqref{P2} and \eqref{P1} when $\epsilon_1, \epsilon_2$ are sufficiently small. In each inner iteration, the complexities of solving the subproblems for updating $\boldsymbol{W}$, $\boldsymbol{X}$, $\boldsymbol{r}\!_{\text{a}}$, and $\boldsymbol{r}\!_{\text{d}}$ are $O(KN^2)$, $O(K^2N)$, $O(K^2N)$, and $O(K^2N)$, respectively, where $O(\cdot)$ denotes the big-$O$ notation. As a result, the overall complexity of Algorithm 2 is $O(T_\text{total}(KN^2+K^2N))$, where $T_\text{total}$ denotes the total number of inner iterations summed over all outer iterations.
\vspace{-0.3cm}
\section{DL Framework for Problem \eqref{P1}}
\begin{figure*}[t]
\centering \includegraphics[width=1.65\columnwidth]{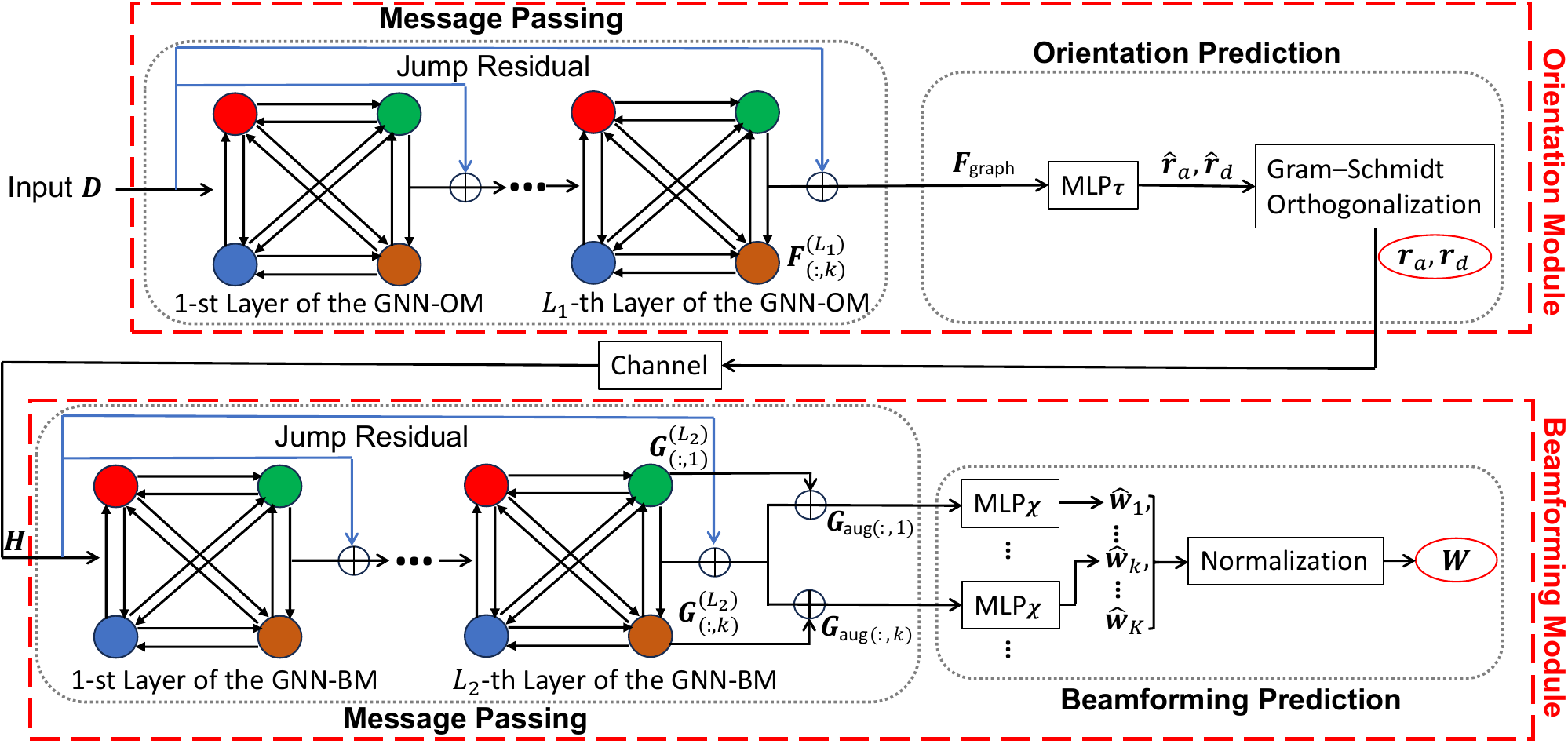} 
\vspace{-0.3cm}
\caption{The architecture of the proposed GNN-based DL framework for joint RAA orientation and beamforming design.}
\vspace{-0.5cm}
\label{GNN}
\end{figure*}
In Section III, an optimization framework based on the PDD method was proposed to achieve a KKT solution of problem \eqref{P1}. However, the convergence of the double-loop iteration may require a long computation time. To accelerate the joint design of RAA orientation and beamforming for UAV-aided communication, in this section, we propose a DL framework to solve \eqref{P1}. Instead of iteratively computing a specific decision of $\{\boldsymbol{r}\!_{\text{a}},\boldsymbol{r}\!_{\text{d}},\boldsymbol{W}\}$ for each configuration of multiuser positions as is done under the optimization framework, the DL framework aims to learn offline, exploiting deep neural networks (DNNs), the rule for mapping all possible user positions to the associated optimal decisions. Then, in online deployment, the predicted $\{\boldsymbol{r}\!_{\text{a}},\boldsymbol{r}\!_{\text{d}},\boldsymbol{W}\}$ are directly output based on the learned mapping function, which requires only a constant number of computations determined by the DNN architecture without any iterations. As a result, the DL framework has the potential to significantly reduce the computation time required for solving \eqref{P1}. However, to approach the performance of the optimization framework with DL, a central challenge lies in how to meticulously design the DNN architecture and training to effectively represent \eqref{P1} and capture its solutions. This problem will be addressed in this section.
\vspace{-0.4cm}
\subsection{GNN-based Architecture of the Proposed DL Framework}
We propose to model the system as a graph and use GNNs to solve \eqref{P1}. This is motivated by the fact that the design of $\{\boldsymbol{r}\!_{\text{a}},\boldsymbol{r}\!_{\text{d}},\boldsymbol{W}\}$ depends on the relative spatial positions of the users and the UAV (cf. Section~II), which naturally form a graph topology. Moreover, for the joint design problem \eqref{P1}, the RAA orientation $\{\boldsymbol{r}\!_{\text{a}},\boldsymbol{r}\!_{\text{d}}\}$ and beamforming matrix $\boldsymbol{W}$ exhibit \emph{invariance} and \emph{equivariance} w.r.t. user permutations, respectively. Specifically, define $\boldsymbol{D}\triangleq[\boldsymbol{u}_1-\boldsymbol{p},...,\boldsymbol{u}_K-\boldsymbol{p}]\in\mathbb{R}^{3\times K}$ as the position matrix of the $K$ users relative to the UAV, and let $\mathcal{\pi}_{m,n}(\boldsymbol{Y})$ denote the permutation matrix that swaps the $m$-th and $n$-th columns of matrix $\boldsymbol{Y}$. Given an optimal solution, denoted by $\{\boldsymbol{r}\!_{\text{a}}^{\star},\boldsymbol{r}\!_{\text{d}}^{\star},\boldsymbol{W}^{\star}\}$, of \eqref{P1} for relative user positions $\boldsymbol{D}$, the permutation invariance and equivariance properties imply that $\{\boldsymbol{r}\!_{\text{a}}^{\star},\boldsymbol{r}\!_{\text{d}}^{\star},\mathcal{\pi}_{m,n}(\boldsymbol{W}^{\star})\}$ constitutes an optimal solution of \eqref{P1} under the permuted positions $\mathcal{\pi}_{m,n}(\boldsymbol{D})$. Compared with other DNNs, GNNs not only well preserve these properties of \eqref{P1}, but also exploit them to improve the efficiency of offline training and enhance the generalization capability during online deployment of the DL framework.

However, unlike the existing GNN-based multiuser beamforming designs in \cite{r17,r18,r19}, a tailored GNN-based structure and training strategy need to be further developed in this paper in order to address the gradient imbalance affecting offline training. As shown in \eqref{phase_effect} and \eqref{SINR}, $\boldsymbol{r}_{\text{a}}$ and $\boldsymbol{r}_{\text{d}}$ reshape the channels through the complex exponential terms and amplitudes in the steering vectors, respectively, while $\boldsymbol{W}$ directly scales the transmit power for each user based on their channels. Thus, $\{\boldsymbol{r}\!_{\text{a}},\boldsymbol{r}\!_{\text{d}}\}$ and $\boldsymbol{W}$ impact the achievable sum-rate via distinct mechanisms and on different scales. When using gradient descent to train the network, the gradient w.r.t. $\boldsymbol{W}$ often has a much larger magnitude than that w.r.t. $\{\boldsymbol{r}\!_{\text{a}},\boldsymbol{r}\!_{\text{d}}\}$, which impedes the optimization of the RAA orientation. To overcome this problem, we develop a two-module GNN architecture that separates the designs of $\{\boldsymbol{r}\!_{\text{a}},\boldsymbol{r}\!_{\text{d}}\}$ and $\boldsymbol{W}$ into two modules, together with a two-stage training strategy that provides flexible control for training both modules.

Fig.~\ref{GNN} depicts the overall architecture of the proposed DL framework, which consists of an orientation module (OM) and a beamforming module (BM) to sequentially output the predicted $\{\boldsymbol{r}\!_{\text{a}},\boldsymbol{r}\!_{\text{d}}\}$ and $\boldsymbol{W}$, respectively. In each module, we model the considered communication system as a fully connected directed graph \cite{r19}, denoted by $\mathcal{G}_\text{OM}$ ($\mathcal{G}_\text{BM}$) for the OM (BM), with $K$ nodes and $K(K-1)$ directed edges. Specifically, $\mathcal{G}_\text{OM}$ and $\mathcal{G}_\text{BM}$ capture each user $k$ as the $k$-th node in the graph, taking user $k$'s relative position $\boldsymbol{D}_{(k,:)}$ and channel vector $\boldsymbol{h}_k$ as the feature of the $k$-th node. Subsequently, GNNs are employed to process the node features of $\mathcal{G}_\text{OM}$ and $\mathcal{G}_\text{BM}$ via the message passing mechanism \cite{GNN}, aiming to extract useful features that are embedded in the input and beneficial for designing $\{\boldsymbol{r}\!_{\text{a}},\boldsymbol{r}\!_{\text{d}}\}$ and $\boldsymbol{W}$, respectively. We refer to the GNN of the OM (BM) as GNN-OM (GNN-BM). Finally, two multilayer perceptrons (MLPs) are used to predict $\{\boldsymbol{r}\!_{\text{a}},\boldsymbol{r}\!_{\text{d}}\}$ and $\boldsymbol{W}$ from the embedded features produced by the GNN-OM and GNN-BM, respectively. The architecture of each module as well as the message passing and network training procedure are further elaborated in Sections~IV--B to~IV--D.
\vspace{-0.4cm}
\subsection{The RAA Orientation Module}
As shown in Fig.~\ref{GNN}, the OM takes node feature matrix $\boldsymbol{D}$ associated with $\mathcal{G}_\text{OM}$ as input and maps it to RAA orientation $\{\boldsymbol{r}\!_{\text{a}},\boldsymbol{r}\!_{\text{d}}\}$ through message passing over an $L_1$-layer GNN-OM, followed by prediction using an MLP. We define the node feature matrix at the $l_1$-th layer of GNN-OM as $\boldsymbol{F}^{(l_1)}\in\mathbb{R}^{F(l_1)\times K}$ for $l_1=0,1,...,L_1$, where $l_1=0$ and $l_1=L_1$ correspond to the input and output layers of GNN-OM, respectively, and $F(l_1)$ denotes the dimension of node features. For our setting, $\boldsymbol{F}^{(0)}=\boldsymbol{D}$ and $F(0)=3$ dimensions. In the following, we present the message passing in GNN-OM and the prediction of $\{\boldsymbol{r}\!_{\text{a}},\boldsymbol{r}\!_{\text{d}}\}$ using the MLP step by step.
\subsubsection{Message Passing in GNN-OM}
In the message passing in the $l_1$-th layer of GNN-OM, the feature $\boldsymbol{F}^{(l_1)}_{(:,k)}$ of the $k$-th node is updated by \emph{aggregating} the features $\boldsymbol{F}^{(l_1\!-\!1)}$ of its neighboring nodes in the previous layer. To distinguish the influence of different neighbors, an attention mechanism is further incorporated into the message passing, enabling the $k$-th node to assign different weights to its neighbors. We assume $A(l_1)$ attention heads are applied in the $l_1$-th layer of GNN-OM, whose outputs are subsequently concatenated via the operator $\mathbin{/\mkern-5mu/}$. The concatenation preserves the distinct features of each head, yielding the aggregated feature $\tilde{\boldsymbol{F}}^{(l_1)}_{(:,k)}\in\mathbb{R}^{F(l_1)\times 1}$ of the $k$-th node in the $l_1$-th layer of GNN-OM given by
\begin{align}
\tilde{\boldsymbol{F}}^{(l_1)}_{(:,k)}=\mathbin{/\mkern-5mu/}_{i=1}^{A(l_1)}(\sum\nolimits_{m=1}^{K}a^{(l_1)}_{k,m,i}\boldsymbol{\Omega}_i^{(l_1)}\boldsymbol{F}^{(l_1-1)}_{(:,m)}).
\label{AGG}
\end{align}
In \eqref{AGG}, $\boldsymbol{\Omega}_i^{(l_1)}\in\mathbb{R}^{\frac{F(l_1)}{A(l_1)}\times F(l_1-1)}$ denotes the learnable matrix of the $i$-th attention head used for the aggregation. $a^{(l_1)}_{k,m,i}$ is the attention weight of the $i$-th attention head for aggregating features at the $k$-th node from the $m$-th node, given as \cite{GNN}
\begin{align}
a^{(l_1)}_{k,m,i}=\frac{ e^{ \boldsymbol{A}^{(l_1)}_{(i,:)}\left(\boldsymbol{\Omega}_i^{(l_1)}\boldsymbol{F}^{(l_1-1)}_{(:,k)}\mathbin{/\mkern-5mu/}\boldsymbol{\Omega}_i^{(l_1)}\boldsymbol{F}^{(l_1-1)}_{(:,m)}\right)} }{\sum\nolimits_{q=1}^{K}e^{ \boldsymbol{A}^{(l_1)}_{(i,:)}\left(\boldsymbol{\Omega}_i^{(l_1)}\boldsymbol{F}^{(l_1-1)}_{(:,k)}\mathbin{/\mkern-5mu/}\boldsymbol{\Omega}_i^{(l_1)}\boldsymbol{F}^{(l_1-1)}_{(:,q)}\right)}},
\end{align}
where $\boldsymbol{A}^{(l_1)}\in\mathbb{R}^{A(l_1)\times \frac{2F(l_1)}{A(l_1)}}$ denotes the learnable matrix for the multi-head attention mechanism, with its row vector $\boldsymbol{A}^{(l_1)}_{(i,:)}$ applied at the $i$-th head. However, due to successive message passing in \eqref{AGG}, GNN-OM may suffer from over-smoothing, where the features of different nodes become increasingly homogeneous, impairing GNN's ability to distinguish among users and degrading the system performance \cite{r19}. To mitigate this, we introduce a jump residual from the input $\boldsymbol{D}$ in each layer of GNN-OM to preserve the original node features. As a result, the feature of the $k$-th node in the $l_1$-th layer of GNN-OM is updated as
\begin{align}
\boldsymbol{F}_{(:,k)}^{(l_1)}=\tilde{\boldsymbol{F}}^{(l_1)}_{(:,k)}+\boldsymbol{\Omega}_\text{jump}^{(l_1)}\boldsymbol{F}^{(0)}_{(:,k)},
\label{F_update}
\end{align}
where $\boldsymbol{\Omega}_\text{jump}^{(l_1)}\!\in\!\mathbb{R}^{F(l_1)\times 3}$ is the learnable matrix for the residual.
\subsubsection{RAA Orientation Prediction}
Unfortunately, directly feeding $\boldsymbol{F}^{(L_1)}$ into an MLP will disrupt the permutation invariance of orientation design, because a permuted input $\mathcal{\pi}_{m,n}(\boldsymbol{D})$ of GNN-OM only leads to equivariantly permuted embedded features $\mathcal{\pi}_{m,n}(\boldsymbol{F}^{(L_1)})$. To tackle this, we instead feed a graph-level representation $\boldsymbol{F}_\text{graph}$ of $\boldsymbol{F}^{(L_1)}$ into the MLP. Specifically, $\boldsymbol{F}_\text{graph}$ is obtained from $\boldsymbol{F}^{(L_1)}$ via a permutation-invariant mapping $\mathcal{P}_{\text{OM}}(\cdot)$, defined as
\begin{align}
\boldsymbol{F}_\text{graph}=\mathcal{P}_{\text{OM}}(\boldsymbol{F}^{(L_1)})=\mathcal{P}_\text{mean}(\boldsymbol{F}^{(L_1)})\mathbin{/\mkern-5mu/}\mathcal{P}_\text{max}(\boldsymbol{F}^{(L_1)}),
\end{align}
where $\boldsymbol{F}_\text{graph}\in\mathbb{R}^{2F(L_1)\times 1}$, $\mathcal{P}_\text{mean}(\boldsymbol{F}^{(L_1)})=\frac{1}{K}\sum\nolimits_{k=1}^{K}\boldsymbol{F}^{(L_1)}_{(:,k)}$, and $\mathcal{P}_{\text{max}}(\boldsymbol{F}^{(L_1)})$ takes the maximum value of each row of $\boldsymbol{F}^{(L_1)}$. This ensures $\boldsymbol{F}_\text{graph}$ to be invariant to user permutations. Finally, an MLP with learnable parameters $\boldsymbol{\tau}$, denoted by $ f_{\boldsymbol{\tau}}(\cdot)$, is employed to predict $\{\boldsymbol{r}\!_{\text{a}},\boldsymbol{r}\!_{\text{d}}\}$ by
\begin{align}
[\hat{\boldsymbol{r}}\!_{\text{a}}^\mathrm{T},\hat{\boldsymbol{r}}\!_{\text{d}}^\mathrm{T}]^\mathrm{T}=f_{\boldsymbol{\tau}}(\boldsymbol{F}_{\text{graph}})\in\mathbb{R}^{6\times 1}.
\label{Oripred}
\end{align}
However, $\hat{\boldsymbol{r}}\!_{\text{a}}$ and $\hat{\boldsymbol{r}}\!_{\text{d}}$ in \eqref{Oripred} are not necessarily orthonormal. To satisfy constraints C2--C4, Gram-Schmidt orthogonalization is adopted at the output of the OM, resulting in
\begin{align}
\boldsymbol{r}\!_{\text{a}}=\frac{\hat{\boldsymbol{r}}\!_{\text{a}}}{\|\hat{\boldsymbol{r}}\!_{\text{a}}\|}, \ \text{}
\boldsymbol{r}\!_{\text{d}}=\frac{\hat{\boldsymbol{r}}\!_{\text{d}}-(\hat{\boldsymbol{r}}\!_{\text{d}}^\mathrm{T}\boldsymbol{r}\!_{\text{a}})\boldsymbol{r}\!_{\text{a}}}{\| \hat{\boldsymbol{r}}\!_{\text{d}}-(\hat{\boldsymbol{r}}\!_{\text{d}}^\mathrm{T}\boldsymbol{r}\!_{\text{a}})\boldsymbol{r}\!_{\text{a}} \|}.
\end{align}
\subsection{The Beamforming Module}
For given $\{\boldsymbol{r}\!_{\text{a}},\boldsymbol{r}\!_{\text{d}}\}$ output from the OM, the resulting channel vector $\boldsymbol{h}_k$ can be directly computed via \eqref{channel}. Since all GNNs in this paper are real-valued, we represent the complex-valued $\boldsymbol{h}_k$ by stacking its real and imaginary parts to preserve the full channel information. Therefore, we define
\begin{align}
\tilde{\boldsymbol{H}}=\begin{bmatrix}
\Re\{\boldsymbol{h}_1\}...\Re\{\boldsymbol{h}_k\}...\Re\{\boldsymbol{h}_K\} \\
\Im\{\boldsymbol{h}_1\}...\Im\{\boldsymbol{h}_k\}... \Im\{\boldsymbol{h}_K\}
\end{bmatrix}\in\mathbb{R}^{2N\times K}
\end{align}
as the node feature matrix of $\mathcal{G}_\text{BM}$. Subsequently, the GNN-BM with $L_2$ layers is employed to extract the embedded features of $\tilde{\boldsymbol{H}}$ using message passing. This is elaborated below along with the prediction of $\boldsymbol{W}$ using an MLP.
\subsubsection{Message Passing in GNN-BM} Let us define the node feature matrix in the $l_2$-th layer of GNN-BM as $\boldsymbol{G}^{(l_2)}\in\mathbb{R}^{G(l_2)\times K}$ for $l_2=0,1,...,L_2$, where $G(l_2)$ denotes the dimension of node features. The input of GNN-BM is $\boldsymbol{G}^{(0)}=\tilde{\boldsymbol{H}}$ with $G(0)=2N$ dimensions. Similar to GNN-OM, the message passing in the $l_2$-th layer of GNN-BM employs $B(l_2)$ attention heads and a jump residual to distinguish the
influence of different neighbors at each node while mitigating the over-smoothing issue. Consequently, the feature of the $k$-th node in the $l_2$-th layer of the GNN-BM is updated as
\begin{align}
\boldsymbol{G}^{(l_2)}_{(:,k)} =\mathbin{/\mkern-5mu/}_{i=1}^{B(l_2)}(\sum\nolimits_{m=1}^{K}b^{(l_2)}_{k,m,i}\boldsymbol{\Phi}_i^{(l_2)}\boldsymbol{G}^{(l_2-1)}_{(:,m)})+\boldsymbol{\Phi}_\text{jump}^{(l_2)}\boldsymbol{G}^{(0)}_{(:,k)},
\label{G_update}
\end{align}
where $\boldsymbol{\Phi}_i^{(l_2)}\in\mathbb{R}^{\frac{G(l_2)}{B(l_2)}\times G(l_2-1)}$ denotes the learnable matrix for aggregation in the $i$-th attention head and $\Phi^{(l_2)}_{\text{jump}}\in\mathbb{R}^{G(l_2)\times 2N}$ denotes the learnable matrix for the jump residual. Define $\boldsymbol{B}^{(l_2)}\in\mathbb{R}^{B(l_2)\times \frac{2G(l_2)}{B(l_2)}}$ as the learnable matrix for the multi-head attention mechanism. The attention weight $b^{(l_2)}_{k,m,i}$ is given by
\begin{align}
b^{(l_2)}_{k,m,i}=\frac{ e^{ \boldsymbol{B}^{(l_2)}_{(i,:)}\left(\boldsymbol{\Phi}_i^{(l_2)}\boldsymbol{G}^{(l_2-1)}_{(:,k)}\mathbin{/\mkern-5mu/}\boldsymbol{\Phi}_i^{(l_2)}\boldsymbol{G}^{(l_2-1)}_{(:,m)}\right)} }{\sum\nolimits_{q=1}^{K}e^{ \boldsymbol{B}^{(l_2)}_{(i,:)}\left(\boldsymbol{\Phi}_i^{(l_2)}\boldsymbol{G}^{(l_2-1)}_{(:,k)}\mathbin{/\mkern-5mu/}\boldsymbol{\Phi}_i^{(l_2)}\boldsymbol{G}^{(l_2-1)}_{(:,q)}\right)}}.
\end{align}
\subsubsection{RAA Beamforming Prediction} As shown in \eqref{SINR}, the optimal $\boldsymbol{w}_{k}^{\star}$ of user $k$ depends not only on its own channel but also on the channels of other users through the interference terms in SINR. Thus, it is insufficient to predict $\boldsymbol{w}_k$ from the embedded features $\boldsymbol{G}^{(L_2)}_{(:,k)}$ of node $k$ alone. Instead, we augment the embedded features $\boldsymbol{G}^{(L_2)}_{(:,k)}$ of each node $k$ with the features of other nodes, resulting in a feature matrix $\boldsymbol{G}_{\text{aug}}\in\mathbb{R}^{3G(L_2)\times K}$, where each column $\boldsymbol{G}_{\text{aug}(:,k)}$ denotes the augmented embedded features of node $k$ and is given by
\begin{align}
\boldsymbol{G}_{\text{aug}(:,k)}=\boldsymbol{G}^{(L_2)}_{(:,k)}\!\mathbin{/\mkern-5mu/}\!\mathcal{P}_\text{mean}(\boldsymbol{G}^{(L_2)})\!\mathbin{/\mkern-5mu/}\!\mathcal{P}_\text{max}(\boldsymbol{G}^{(L_2)}).
\label{aug}
\end{align}
Note that, with \eqref{aug}, $\boldsymbol{G}_{\text{aug}}$ is equivariant to user permutations, preserving the permutation equivariance of the beamforming design. Finally, we apply an MLP $f_{\boldsymbol{\chi}}(\cdot)$ with learnable parameters $\boldsymbol{\chi}$ to each node $k$ separately by taking $\boldsymbol{G}_{\text{aug}(:,k)}$ as the input, and predict $\boldsymbol{w}_k$ via
\begin{align}
\begin{bmatrix}
\Re\{\hat{\boldsymbol{w}}_k \} \\
\Im\{\hat{\boldsymbol{w}}_k \}
\end{bmatrix}
=f_{\boldsymbol{\chi}}(\boldsymbol{G}_{\text{aug}(:,k)})\in\mathbb{R}^{2N\times 1}.
\end{align}
Define $\hat{\boldsymbol{W}}\!\triangleq\![\hat{\boldsymbol{w}}_1,...,\hat{\boldsymbol{w}}_K]\!\in\!\mathbb{R}^{N\!\times\! K}$. As $\hat{\boldsymbol{W}}$ may not satisfy constraint C1, a normalization layer is employed at the output of BM, resulting in
\begin{align}
\boldsymbol{W}=
\begin{cases}
{\hat{\boldsymbol{W}}} / {\| \hat{\boldsymbol{W}} \|_{\text{F}}}, & \text{if } \| \hat{\boldsymbol{W}} \|_{\text{F}}^2 > P,\\
\hat{\boldsymbol{W}},  & \text{otherwise}.
\end{cases}
\end{align}
\vspace{-0.4cm}
\subsection{The Two-stage Training Strategy}
For training the proposed GNN-based DL architecture, it would be intractable to obtain the optimal $\{\boldsymbol{r}\!_{\text{a}}^{\star},\boldsymbol{r}\!_{\text{d}}^{\star},\boldsymbol{W}^{\star}\}$ and use them as supervision labels. In fact, even the near-optimal solutions from Section~III are computationally costly to generate for training. For these reasons, we adopt unsupervised learning to train the OM and BM for maximizing the sum-rate subject to constraints C1--C4 and C6 without labels. While constraints C1 and C2--C4 have been easily handled by the normalization in BM and the orthogonalization in OM, respectively, the QoS constraint C6 is difficult to enforce through the module architecture itself. Instead, we incorporate C6 into the loss function as a penalty term via the Lagrangian duality method (LDM) \cite{r19}, which adaptively balances the sum-rate maximization and the satisfaction of C6 by updating the dual variable jointly with the learnable parameters during training. We use $\boldsymbol{\Omega}_{\text{OM}}$ and $\boldsymbol{\Phi}_{\text{BM}}$ to represent all learnable parameters of the OM and BM, respectively. Training the whole framework is equivalent to solving the following problem
\begin{equation}
\begin{aligned}
\underset{\boldsymbol{\mu}}{\text{max}} \underset{\boldsymbol{\Omega}_{\text{OM}}, \boldsymbol{\Phi}_{\text{BM}}}{\text{min}}\mathcal{L}_\text{training}=&-\sum\nolimits_{k=1}^{K}R_k \\
&+\sum\nolimits_{k=1}^{K}\mu_k\text{ReLU}(\overline{\gamma}_k-\gamma_k),
\end{aligned}
\label{P10}
\tag{P10}
\end{equation}
where $\boldsymbol{\mu}\triangleq[\mu_1,...,\mu_k,...,\mu_K]\in\mathbb{R}_{+}^{K\times 1}$ denotes the dual variable associated with C6, and $\text{ReLU}(\cdot)$ represents the rectified linear unit function with $\text{ReLU}(x)\triangleq\max(0,x)$. $\boldsymbol{\mu}$ and $\{ \boldsymbol{\Omega}_{\text{OM}}, \boldsymbol{\Phi}_{\text{BM}} \}$ are alternately updated during offline training.

As discussed in Section~IV--A, directly training the whole framework suffers from the gradient imbalance, where the optimization of $\boldsymbol{\Phi}_{\text{BM}}$ dominates that of $\boldsymbol{\Omega}_{\text{OM}}$. However, since $\boldsymbol{\Omega}_{\text{OM}}$ and $\boldsymbol{\Phi}_{\text{BM}}$ reside in two sequential modules, we can develop a two-stage training strategy, consisting of an OM pre-training stage followed by a joint training stage, to flexibly control the training of each module and mitigate gradient imbalance. The pre-training of the OM, which is independent of the BM, also obtains a heuristic yet effective initialization for $\boldsymbol{\Omega}_{\text{OM}}$, before both modules are jointly trained to maximize the sum-rate.
\subsubsection{Orientation Module Pre-training Stage} In this stage, only $\boldsymbol{\Omega}_{\text{OM}}$ is optimized, and the loss $\mathcal{L}_\text{training}$ cannot be used, since it depends on $\boldsymbol{W}$. Following the RAA orientation initialization strategy in the optimization framework, we minimize instead the pairwise inter-user channel correlation and formulate the pre-training problem as
\begin{equation}
\begin{aligned}
\underset{\boldsymbol{\Omega}_{\text{OM}}}{\text{min}} \ \mathcal{L}_\text{pre-training}=\sum\nolimits_{k=1}^{K}\sum\nolimits_{m=k+1}^{K}|\tilde{\boldsymbol{a}}_k^T\tilde{\boldsymbol{a}}_m|^2.
\end{aligned}
\label{P11}
\tag{P11}
\end{equation}
\subsubsection{Joint Training Stage} In this stage, we use the $\boldsymbol{\Omega}_{\text{OM}}$ obtained in the pre-training stage as the initialization for OM, and then jointly optimize both $\boldsymbol{\Omega}_{\text{OM}}$ and $\boldsymbol{\Phi}_{\text{BM}}$ to maximize the sum-rate as formulated in \eqref{P10}. The overall two-stage training procedure is outlined in Algorithm 3, where $\nabla_{\boldsymbol{\mu}}$ denotes the gradient accumulator for the dual variable.
\renewcommand{\algorithmicrequire}{\textbf{Input:}}
\begin{algorithm}[!t]
\footnotesize
\caption{Two-stage Training Strategy for Solving \eqref{P10}}
\begin{algorithmic}[1]
\REQUIRE Initial $\boldsymbol{\Omega}_{\text{OM}}^{(0)}$, $\boldsymbol{\Phi}_{\text{BM}}^{(0)}$, $\boldsymbol{\mu}^{(0)}$, step size $c$ for updating $\boldsymbol{\mu}$, training dataset $\mathcal{D}=\{\boldsymbol{D}\}_{i=1}^{M_1}$, number of mini-batches $B$, mini-batch size $M_2$, numbers of pre-training and joint-training epochs $E_1$ and $E_2$, respectively.
\STATE \textbf{Orientation Module Pre-training:}
\STATE Initialize $\boldsymbol{\Omega}_{\text{OM}}\leftarrow\boldsymbol{\Omega}_{\text{OM}}^{(0)}$
\FOR{epoch $e_1 = 1, \dots, E_1$}
\FOR{batch $b = 1, \dots, B$}
\STATE Sample the $b$-th mini-batch $\{\boldsymbol{D}\}_{i=1}^{M_2}$ from $\mathcal{D}$.
\STATE Compute the mini-batch loss $\frac{1}{M_2}\sum\nolimits_{i=1}^{M_2}\mathcal{L}_{\text{pre-training},i}$.
\STATE Update $\boldsymbol{\Omega}_{\text{OM}}$ with stochastic gradient descent.
\ENDFOR
\ENDFOR
\STATE Obtain the pre-trained parameters $\tilde{\boldsymbol{\Omega}}_{\text{OM}}$.
\STATE \textbf{Joint Training:}
\STATE Initialize $\boldsymbol{\Omega}_{\text{OM}}\leftarrow\tilde{\boldsymbol{\Omega}}_{\text{OM}}$, $\boldsymbol{\Phi}_{\text{BM}}\leftarrow\boldsymbol{\Phi}_{\text{BM}}^{(0)}$, and $\boldsymbol{\mu}\leftarrow\boldsymbol{\mu}^{(0)}$.
\FOR{epoch $e_2 = 1, \dots, E_2$}
\STATE Initialize $\nabla_{\boldsymbol{\mu}}=[\nabla_{\mu_{1}},...,\nabla_{\mu_{K}}]\leftarrow\boldsymbol{0}$.
\FOR{batch $b = 1, \dots, B$}
\STATE Sample the $b$-th mini-batch $\{\boldsymbol{D}\}_{i=1}^{M_2}$ from $\mathcal{D}$.
\STATE Compute the mini-batch loss $\frac{1}{M_2}\sum\nolimits_{i=1}^{M_2}\mathcal{L}_{\text{training},i}$.
\STATE Jointly update $\boldsymbol{\Omega}_{\text{OM}}$ and $\boldsymbol{\Phi}_{\text{BM}}$ with stochastic gradient descent.
\STATE $\nabla_{\mu_{k}}\leftarrow\nabla_{\mu_{k}}+\frac{1}{M_2}\sum\nolimits_{i=1}^{M_2}\text{ReLU}(\overline{\gamma}_k-\gamma_{k,i})$, for $k\in\mathcal{K}$.
\ENDFOR
\STATE $\boldsymbol{\mu}\leftarrow \boldsymbol{\mu}+c\nabla_{\boldsymbol{\mu}}$.
\ENDFOR
\end{algorithmic}
\end{algorithm}
Note that our proposed GNN-based DL framework preserves the permutation invariance and equivariance of the joint design problem \eqref{P1}, which improves the efficiency of offline training and the generalization during online deployment of the DL framework \cite{GNN}. Additionally, it enables the proposed DL framework to support a varying number of users without changing the dimensions of the learnable parameters. 
\vspace{-0.35cm}
\section{Simulation Results}
\vspace{-0.05cm}
In this section, we evaluate the performance of RAA-based UAV-aided communication using the proposed optimization and DL frameworks by simulation. We equip the UAV with an RAA composed of $N=16$ half-wavelength dipole antennas to serve $K=4$ terrestrial users, where each user employs a single isotropic receive antenna. The users are uniformly and randomly distributed within a square area of $100 \ \text{m}\times 100 \ \text{m}$, centered at the origin $[0,0,0]^\text{T}$ m, and the UAV hovers at position $\boldsymbol{p}=[0,0,40]^\text{T}$ m to ensure reliable LoS links for all users within the entire service area. The total transmit power is set to $P=1$~W and the effective noise power spectral density is assumed to be $-160$ dBm/Hz over a bandwidth of $10$ MHz \cite{r14,r19}. Finally, the required minimum achievable rate for each user is set to $\overline{R}_k=4$ nat/s/Hz for all $k\in \mathcal{K}$.

In the proposed optimization framework, the relevant algorithmic hyperparameters are set as $\epsilon_1=\epsilon_2=\epsilon_4=10^{-6}$, $\epsilon_3=10^{-4}$, $c_1=0.95$, $c_2=0.8$, $T_{\text{in}}^{\text{max}}=T_{\text{in}}^{\text{in}}=300$, $\rho_1^{(0)}=\rho_2^{(0)}=0.5$, $z_2^{(0)}=0$, and $\boldsymbol{Z}_1^{(0)}=\mathbf{0}$. In the proposed DL framework, the learning rate is initialized
to $10^{-5}$ for both pre-training and joint training, and the step size $c$ of updating $\boldsymbol{\mu}$ is set to 0.02. The batch size $B$ is set to 32 for $E_1=500$ epochs of pre-training and $E_2=100$ epochs of joint training. The sizes of the training and validation sets are chosen as 100,000 and 1,0000, respectively. Table~\ref{parameters} shows the implementation details of OM and BM. All simulations are conducted in Python 3.10.19 with PyTorch 2.9.0 on a workstation equipped with an Intel(R) Core(TM) i9-13900HX CPU and an NVIDIA RTX 4060 GPU.
\begin{table}[!t]
\begin{center}
\caption{Architecture of the Proposed DL Framework.}
\label{parameters}
\begin{tabular}{| c | c | c | c | c | c |}
\hline
\multicolumn{3}{|c|}{Orientation Module} & \multicolumn{3}{c|}{Beamforming Module} \\
\hline
Name & Dimension & \#AHs & Name & Dimension & \#AHs \\
\hline
$\boldsymbol{\Omega}_i^{(1)}$ & $64\times 3$ & 10 & $\boldsymbol{\Phi}_i^{(1)}$ & $32\times 32$ & 10 \\
\hline
$\boldsymbol{A}^{(1)}$ & $10 \times 128$ & \diagbox{}{} &  $\boldsymbol{B}^{(1)}$ & $10 \times 64$ & \diagbox{}{}\\
\hline
$\boldsymbol{\Omega}_\text{jump}^{(1)}$ & $640\times 3$ & \diagbox{}{} & $\boldsymbol{\Phi}_\text{jump}^{(1)}$ & $320\times 32$ & \diagbox{}{} \\
\hline
\multicolumn{3}{|c|}{\multirow{6}{*}{\diagbox{}{}}} & $\boldsymbol{\Phi}_i^{(2)}$ & $64\times 320$ & 10 \\
\cline{4-6}
\multicolumn{3}{|c|}{} &
$\boldsymbol{B}^{(2)}$ & $10 \times 128$ & \diagbox{}{}\\
\cline{4-6}
\multicolumn{3}{|c|}{} &
$\boldsymbol{\Phi}_\text{jump}^{(2)}$ & $640\times 32$ & \diagbox{}{} \\
\cline{4-6}
\multicolumn{3}{|c|}{} &
$\boldsymbol{\Phi}_i^{(3)}$ & $128\times 640$ & 10 \\
\cline{4-6}
\multicolumn{3}{|c|}{} &
$\boldsymbol{B}^{(3)}$ & $10 \times 256$ & \diagbox{}{}\\
\cline{4-6}
\multicolumn{3}{|c|}{} &
$\boldsymbol{\Phi}_\text{jump}^{(3)}$ & $1280\times 32$ & \diagbox{}{} \\
\hline
\multirow{3}{*}{$\text{MLP}_{\boldsymbol{\tau}}$} & $1280\times1024$ & \diagbox{}{} & \multirow{4}{*}{$\text{MLP}_{\boldsymbol{\chi}}$} & $3840\times 1024$ & \diagbox{}{}\\
\cline{2-3} \cline{5-6}
 & $1024\times512$ & \diagbox{}{} & & $1024\times512$ & \diagbox{}{}\\
\cline{2-3} \cline{5-6}
 & $512\times6$ & \diagbox{}{} & & $512\times256$ & \diagbox{}{} \\
\cline{1-3}\cline{5-6}
\multicolumn{3}{|c|}{\diagbox{}{}} & & $256\times32$ & \diagbox{}{}\\ 
\hline
\end{tabular}
\\[6pt]
\parbox{0.9\linewidth}{
\footnotesize
\#AHs: Number of attention heads.
}
\vspace{-0.4cm}
\end{center}
\end{table}
\begin{table}[!t]
\begin{center}
\caption{Overview of the Schemes Considered for Simulation.}
\label{schemes}
\begin{tabular}{| c | c | c | c | c | c | c |}
\hline
Schemes & RAA & FAA & Dip & Iso & Opt & DL \\
\hline
\textbf{RAA-Dip-Opt} & $\checkmark$ & & $\checkmark$ & &$\checkmark$ & \\
\hline
\textbf{RAA-Dip-DL} & $\checkmark$ & & $\checkmark$ & & &$\checkmark$ \\
\hline
RAA-Iso-Opt & $\checkmark$ & &  &$\checkmark$ &$\checkmark$ & \\
\hline
RAA-Iso-DL & $\checkmark$ & &  &$\checkmark$ & & $\checkmark$\\
\hline
FAA-Dip-Opt & &$\checkmark$ & $\checkmark$ & &$\checkmark$ & \\
\hline
FAA-Dip-DL & &$\checkmark$ & $\checkmark$ & & & $\checkmark$\\
\hline
FAA-Iso-Opt & &$\checkmark$ &  & $\checkmark$ &$\checkmark$ & \\
\hline
FAA-Iso-DL & &$\checkmark$ &  & $\checkmark$ & & $\checkmark$\\
\hline
\end{tabular}
\\[6pt]
\parbox{0.9\linewidth}{
\footnotesize
Dip: Dipole antenna. \ \ Iso: Isotropic antenna. \\
Opt: Optimization framework. \ \ DL: DL framework.
}
\vspace{-0.4cm}
\end{center}
\end{table}
\begin{table}[!t]
\begin{center}
\caption{Average Computation Time in Seconds of RAA-Dip-Opt and RAA-Dip-DL for Different Number of Antennas .}
\label{TIME}
\begin{tabular}{| c | c | c | c | c |}
\hline
$N$ & 8 & 16 & 24 & 32 \\
\hline
RAA-Dip-Opt & 4.3946 & 5.3123 & 6.4419 & 7.4846 \\ 
\hline
RAA-Dip-DL & 0.0032 & 0.0035 & 0.0035 & 0.0040\\ 
\hline
\end{tabular}
\vspace{-0.3cm}
\end{center}
\end{table}

For performance comparison, we consider both RAAs and FAAs composed of either dipoles or isotropic antennas for UAV-aided communication using the proposed optimization and DL frameworks\footnote{When using the FAAs, the orientation optimization in Algorithm 2 is omitted, and in the DL framework, the orientation module is removed and we only train the beamforming module without the pre-training stage.}, yielding eight schemes as outlined in Table~\ref{schemes}. RAA-Dip-Opt and RAA-Dip-DL are the proposed schemes as highlighted in bold, while the remaining schemes serve as baselines. For all FAAs, $\boldsymbol{r}\!_{\text{a}}$ is fixed to $[1,0,0]^\text{T}$, and $\boldsymbol{r}\!_{\text{d}}$ is further set to $[0,1,0]^\text{T}$ for the FAA composed of dipoles.
\vspace{-0.8cm}
\subsection{Convergence Analysis and Computation Time}
\vspace{-0.1cm}
\label{CAnalysis}
\begin{figure}[t]
\centering
\subfigure[]{\label{OValue}\includegraphics[width=0.49\columnwidth]{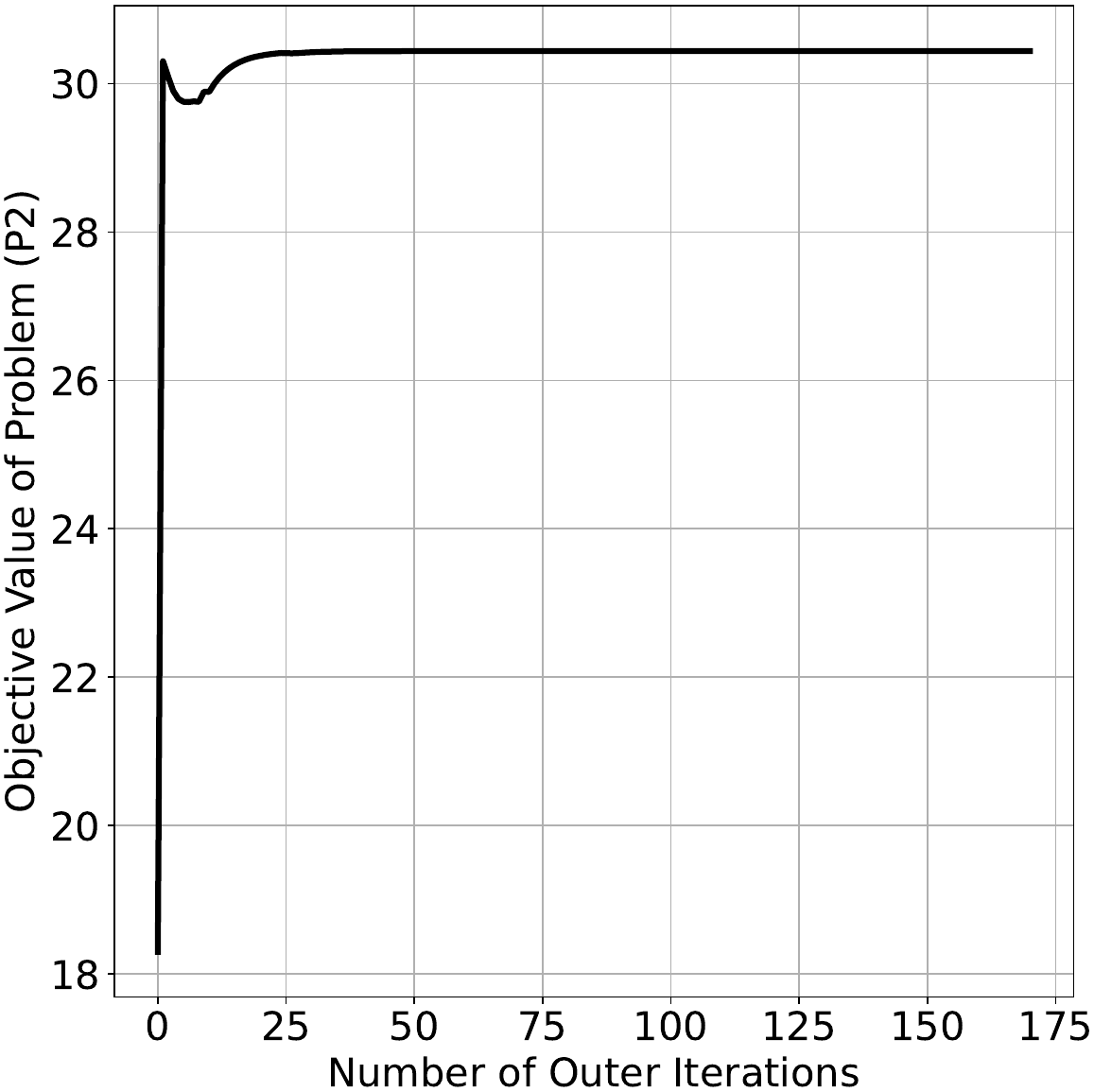}}
\subfigure[]{\label{CViolation}\includegraphics[width=0.49\columnwidth]{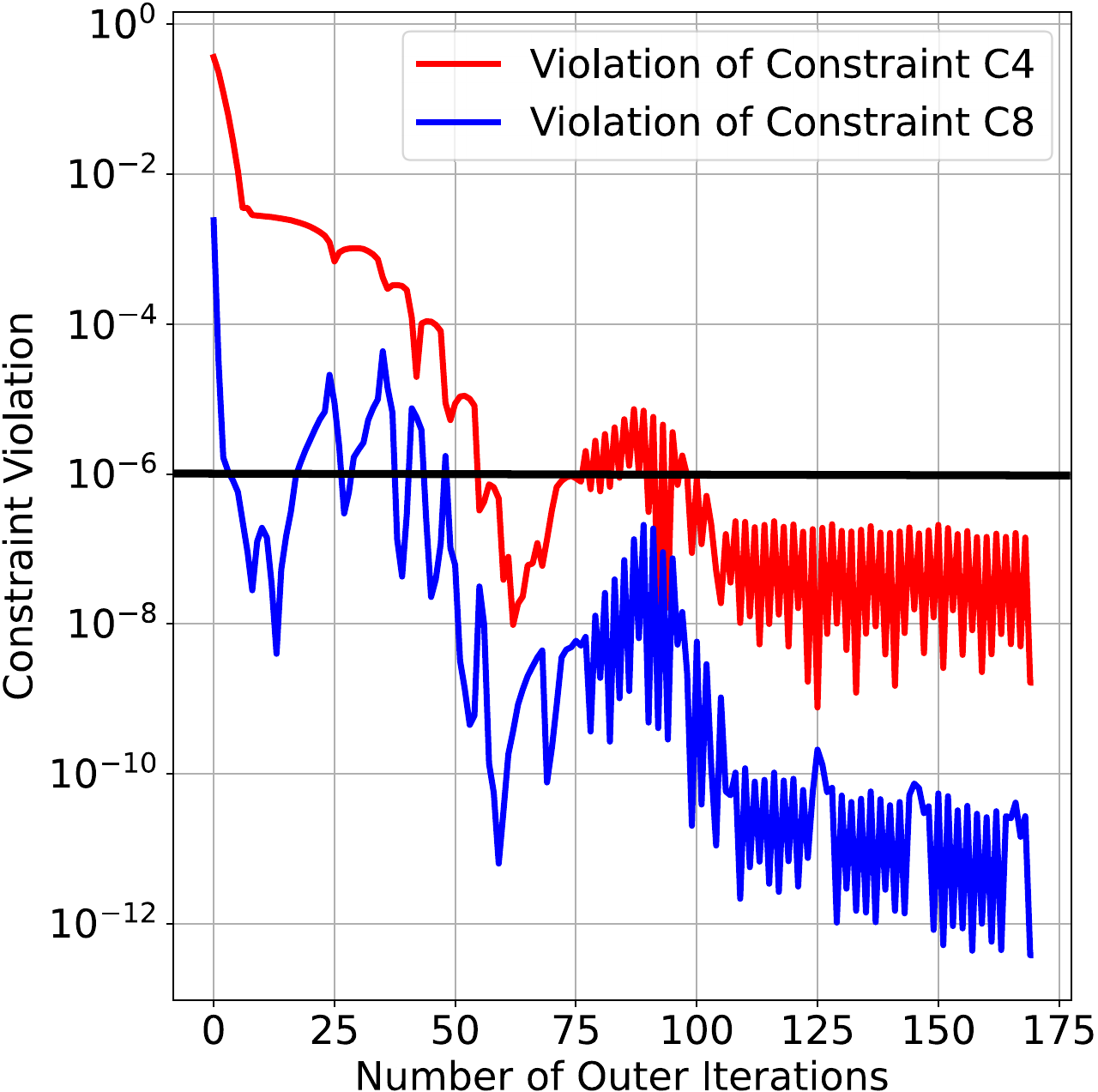}}
\vspace{-0.4cm}
\caption{(a) Objective value and (b) constraint violation of problem \eqref{P2} versus the number of outer iterations in Algorithm 2.
}
\vspace{-0.3cm}
\label{Convergence}
\end{figure}
\begin{figure}[t]
\centering
\subfigure[]{\label{PreTraining}\includegraphics[width=0.49\columnwidth]{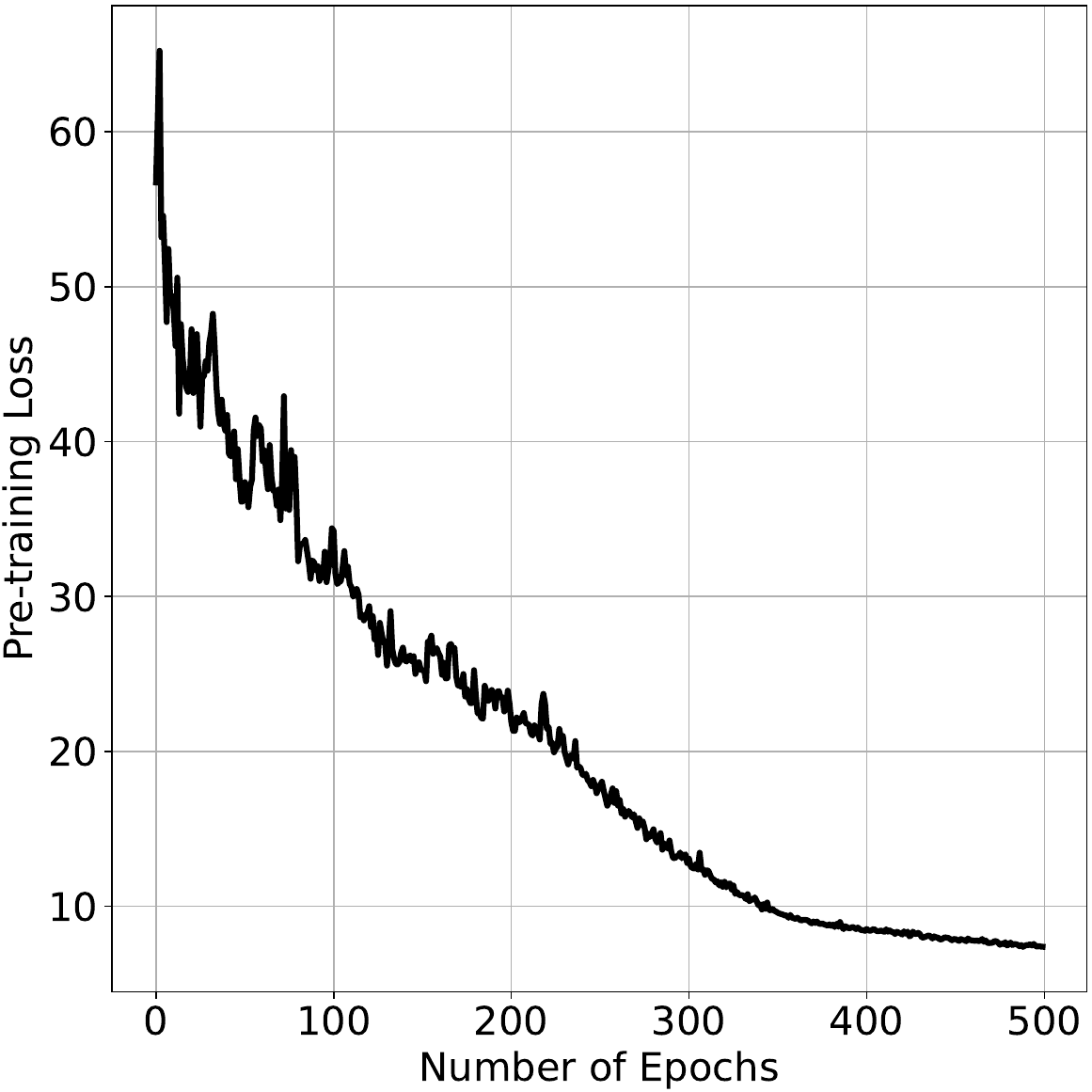}}
\subfigure[]{\label{JointTraining}\includegraphics[width=0.49\columnwidth]{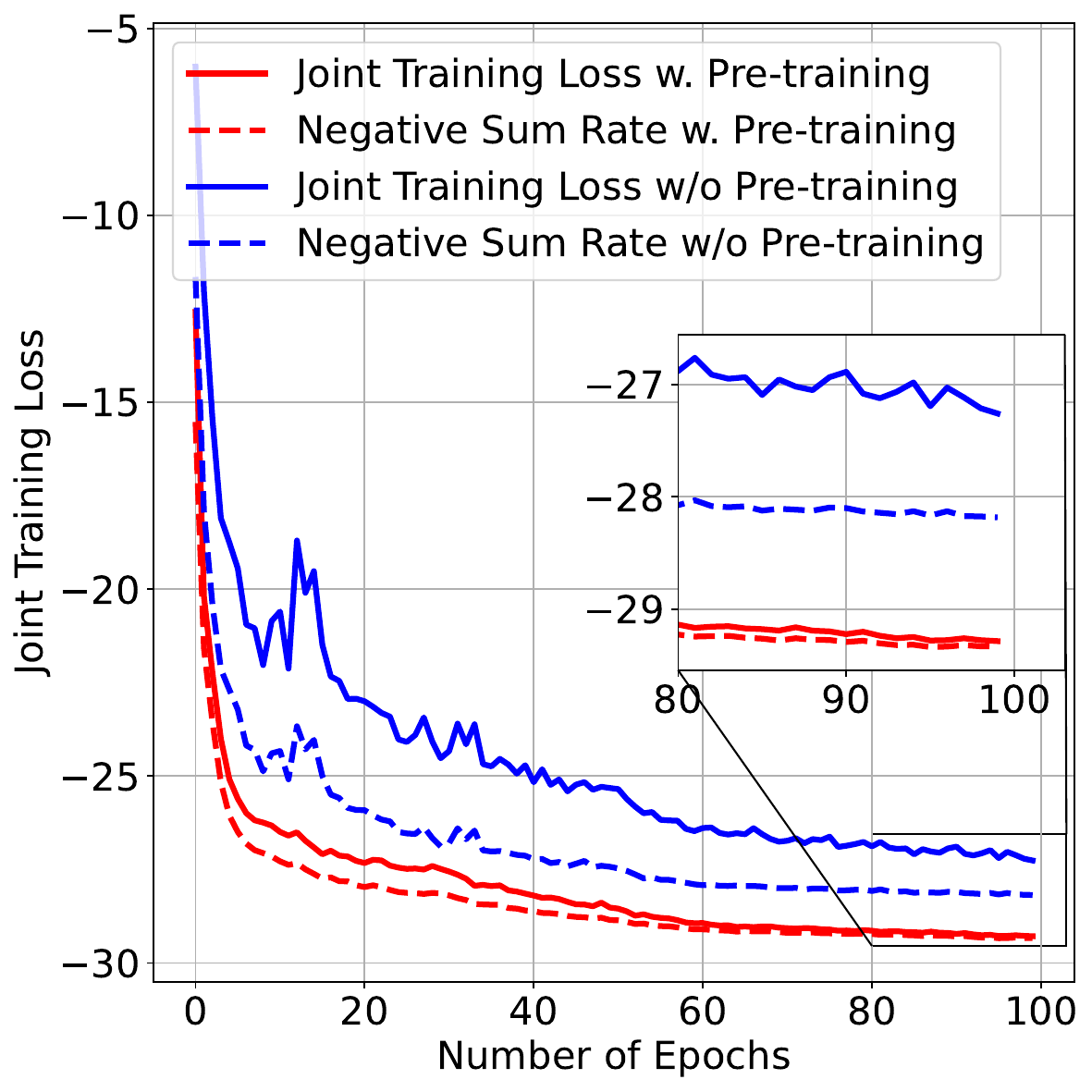}}
\vspace{-0.4cm}
\caption{(a) Pre-training loss and (b) joint training loss of Algorithm 3.
}
\label{Loss}
\end{figure}
\begin{figure}[t]
\centering
\subfigure[]{\label{FAAGain}\includegraphics[width=0.49\columnwidth]{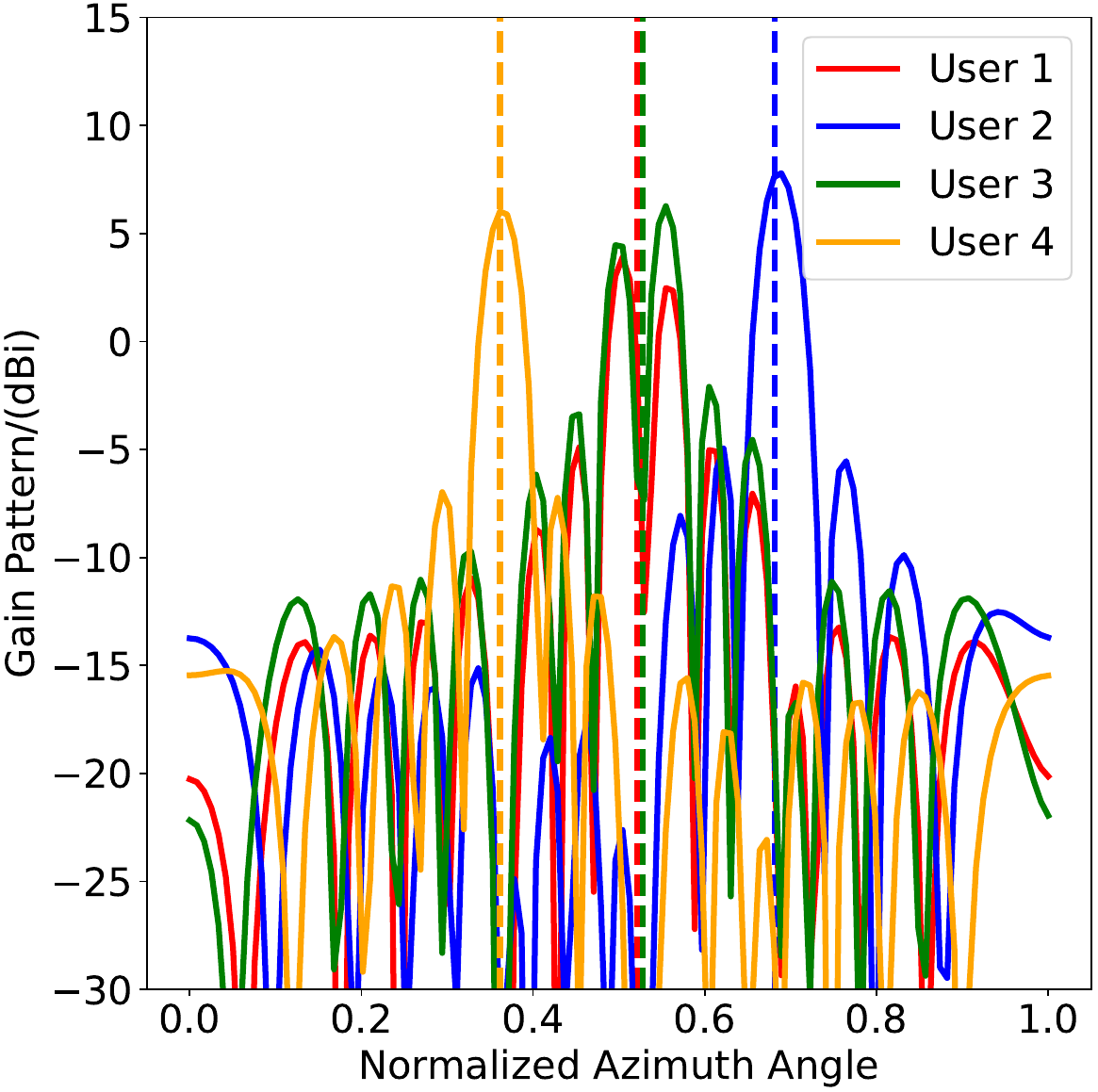}}
\subfigure[]{\label{RAAGain}\includegraphics[width=0.49\columnwidth]{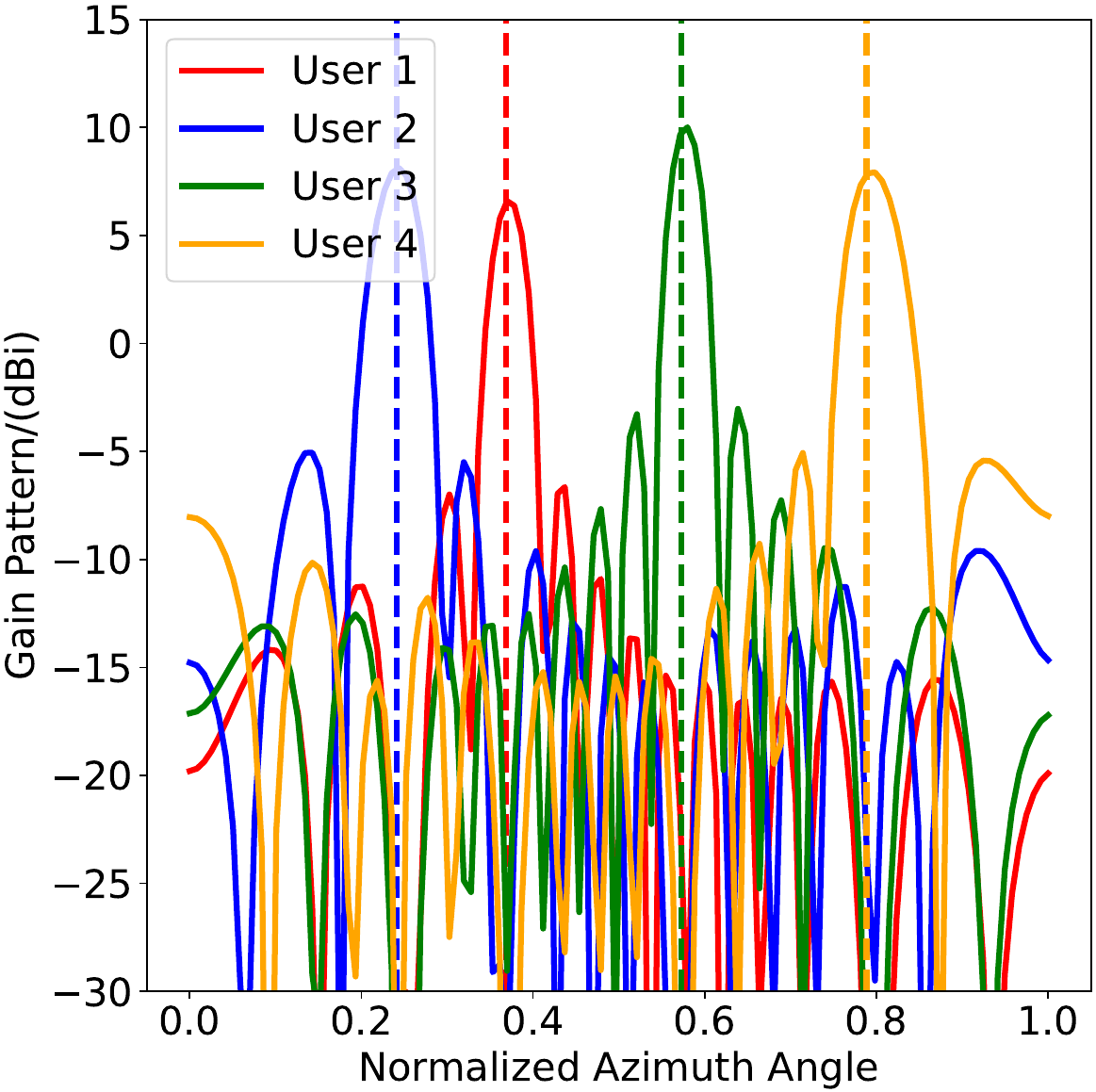}}
\vspace{-0.5cm}
\caption{Transmit beampattern gains of schemes (a) FAA-Dip-Opt and (b) RAA-Dip-Opt.
}
\label{Gain}
\end{figure}
To verify the convergence of Algorithm 2 for the proposed RAA-Dip-Opt, we consider the following illustrative user position setting: $\boldsymbol{u}_1\!=\![-3.73,-39.20,0]^\text{T}$ m, $\boldsymbol{u}_2\!=\![-35.56,-38.56,0]^\text{T}$ m, $\boldsymbol{u}_3\!=\![-3.78,16.69,0]^\text{T}$ m, and $\boldsymbol{u}_4\!=\![28.31,45.97,0]^\text{T}$ m. Fig.~\ref{Convergence} shows the objective $\mathcal{L}_\text{PDD}$ of problem \eqref{P2} and the violations of constraints C4 (i.e., $|\boldsymbol{r}\!_{\text{a}}^T\boldsymbol{r}\!_{\text{d}}|^2$) and C8 (i.e., $\| \boldsymbol{X}-\boldsymbol{H}^H\boldsymbol{W} \|_{\text{F}}^2$) versus the number of outer iterations. From Fig.~\ref{OValue} we observe that, as the iterations continue, $\mathcal{L}_\text{PDD}$ increases rapidly, followed by a slight decrease caused by the reduction of $\rho_1$ and $\rho_2$, when stronger penalties are imposed for C4 and C8. Subsequently, $\mathcal{L}_\text{PDD}$ resumes increasing and converges quickly. From Fig.~\ref{CViolation}, we find that the violation of C4 exceeds that of C8, indicating that optimizing $\{\boldsymbol{r}\!_{\text{a}},\boldsymbol{r}\!_{\text{d}}\}$ is more challenging than optimizing $\boldsymbol{W}$, because $\boldsymbol{r}\!_{\text{a}}$ and $\boldsymbol{r}\!_{\text{d}}$ lie on highly nonconvex sphere manifolds. However, with the proposed RCG method, the violation of C4 become smaller than $10^{-6}$ after approximately 100 iterations.

Fig.~\ref{Loss} shows the pre-training loss $\mathcal{L}_\text{pre-training}$, the joint training loss $\mathcal{L}_\text{training}$ (in solid lines), and the negative sum-rate given by $\mathcal{L}_\text{training}-\sum\nolimits_{k=1}^{K}\mu_k\text{ReLU}(\overline{\gamma}_k-\gamma_k)$ (in dashed lines) of Algorithm 3 for the proposed RAA-Dip-DL. To validate the advantages of the two-stage training strategy, we conduct an ablation study where only the joint training of the OM and BM modules is performed without pre-training. Fig.~\ref{PreTraining} shows that $\mathcal{L}_{\text{pre-training}}$, which represents the channel correlations among the users, decreases dramatically within 500 pre-training epochs. Meanwhile, from Fig.~\ref{JointTraining} we observe that the proposed two-stage training strategy achieves a much lower joint training loss $\mathcal{L}_\text{training}$ than its counterpart without pre-training, resulting in also a higher sum-rate. Moreover, pre-training significantly reduces the gap between the joint training loss and negative sum-rate, indicating that the proposed two-stage training strategy results in a much smaller violation of constraint C6 and is able to better satisfy the QoS of each user. This is because directly training the entire framework suffers severely from the gradient imbalance issue.

Table~\ref{TIME} further compares the average online computation time of the proposed RAA-Dip-Opt and RAA-Dip-DL for different numbers of transmit antennas $N$. As expected, the computation time increases with $N$, because the complexity depends on $N$. We observe that the DL framework entails a much lower computation time than the optimization framework, demonstrating its great potential to support real-time RAA-based UAV-aided communications.
\vspace{-0.6cm}
\subsection{Performance Comparison}
\vspace{-0.1cm}
To demonstrate the advantages of RAAs compared to FAAs in UAV-aided communication, Fig.~\ref{Gain} presents the transmit beampattern gains $|\boldsymbol{w}_{k}^{H}\boldsymbol{a}_{k}|^{2}$, for user $k=1,2,3,4$, achieved by FAA-Dip-Opt and the proposed RAA-Dip-Opt for the same illustrative user position setting as considered in Section~\ref{CAnalysis}. The dashed lines denote the normalized azimuth angles $\varphi_k/\pi$ of the users w.r.t. the array. From Fig.~\ref{FAAGain} we observe that user 1 and 3 suffer from severe interference for FAA-Dip-Opt due to their strong channel correlation, resulting in a sum-rate of 24.29 nat/s/Hz. However, through orientation design, the proposed RAA-Dip-Opt is able to separate users in azimuth angles (cf. Fig.~\ref{RAAGain}) and mitigate the interference. Moreover, the proposed RAA-Dip-Opt attains larger transmit beampattern gains than FAA-Dip-Opt, because it can align the users with the high-gain regions of dipole antennas, resulting in a higher sum-rate of 30.43 nat/s/Hz.
\begin{figure}[t]
\centering
\subfigure[]{\label{VSQoS}\includegraphics[width=0.49\columnwidth]{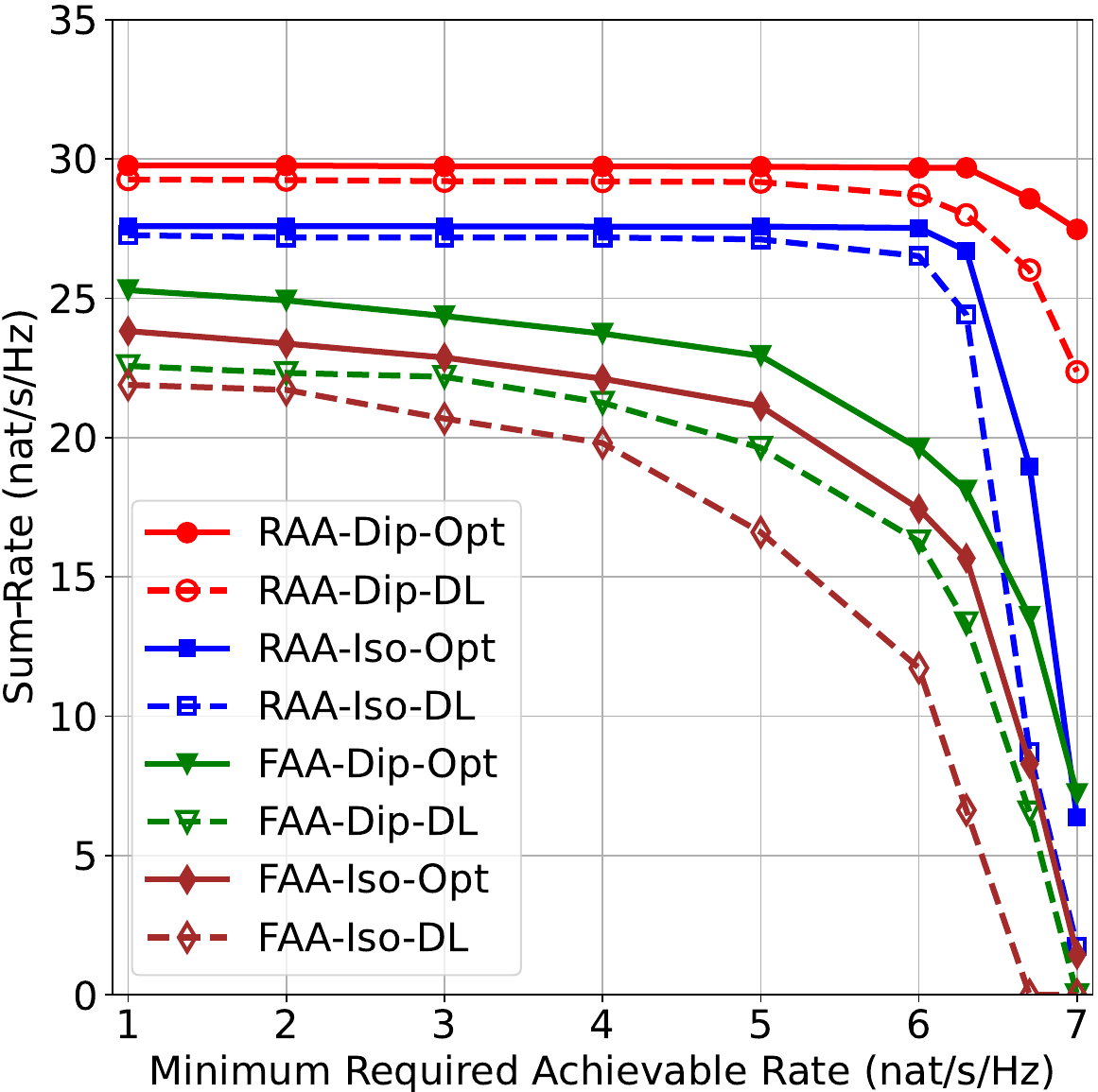}}
\subfigure[]{\label{FVSQoS}\includegraphics[width=0.49\columnwidth]{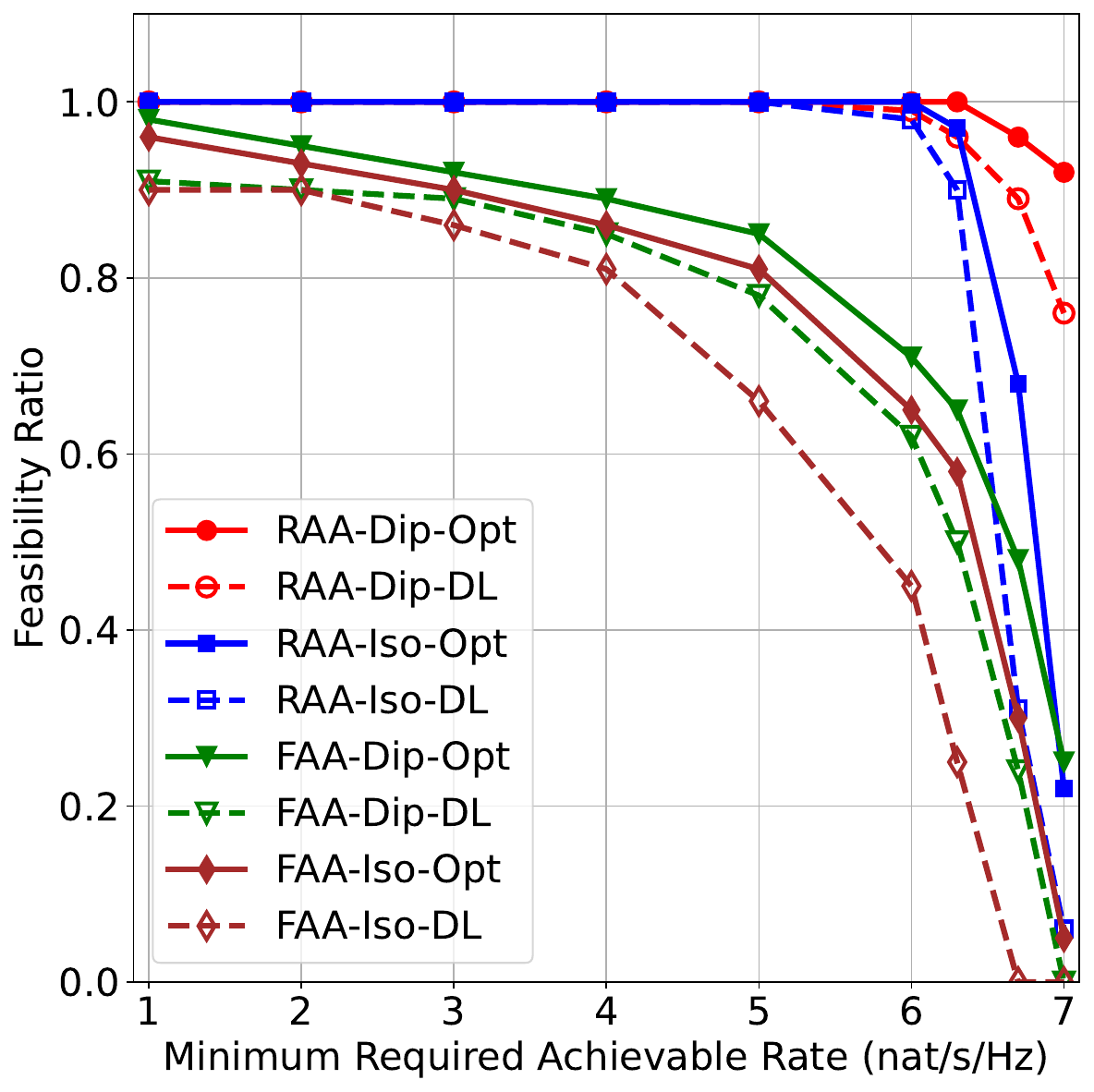}}
\vspace{-0.5cm}
\caption{(a) Sum-rate and (b) feasibility ratio versus QoS requirements.
}
\vspace{-0.1cm}
\label{increase_QoS}
\end{figure}
\begin{figure}[t]
\centering
\subfigure[]{\label{VSantenna}\includegraphics[width=0.49\columnwidth]{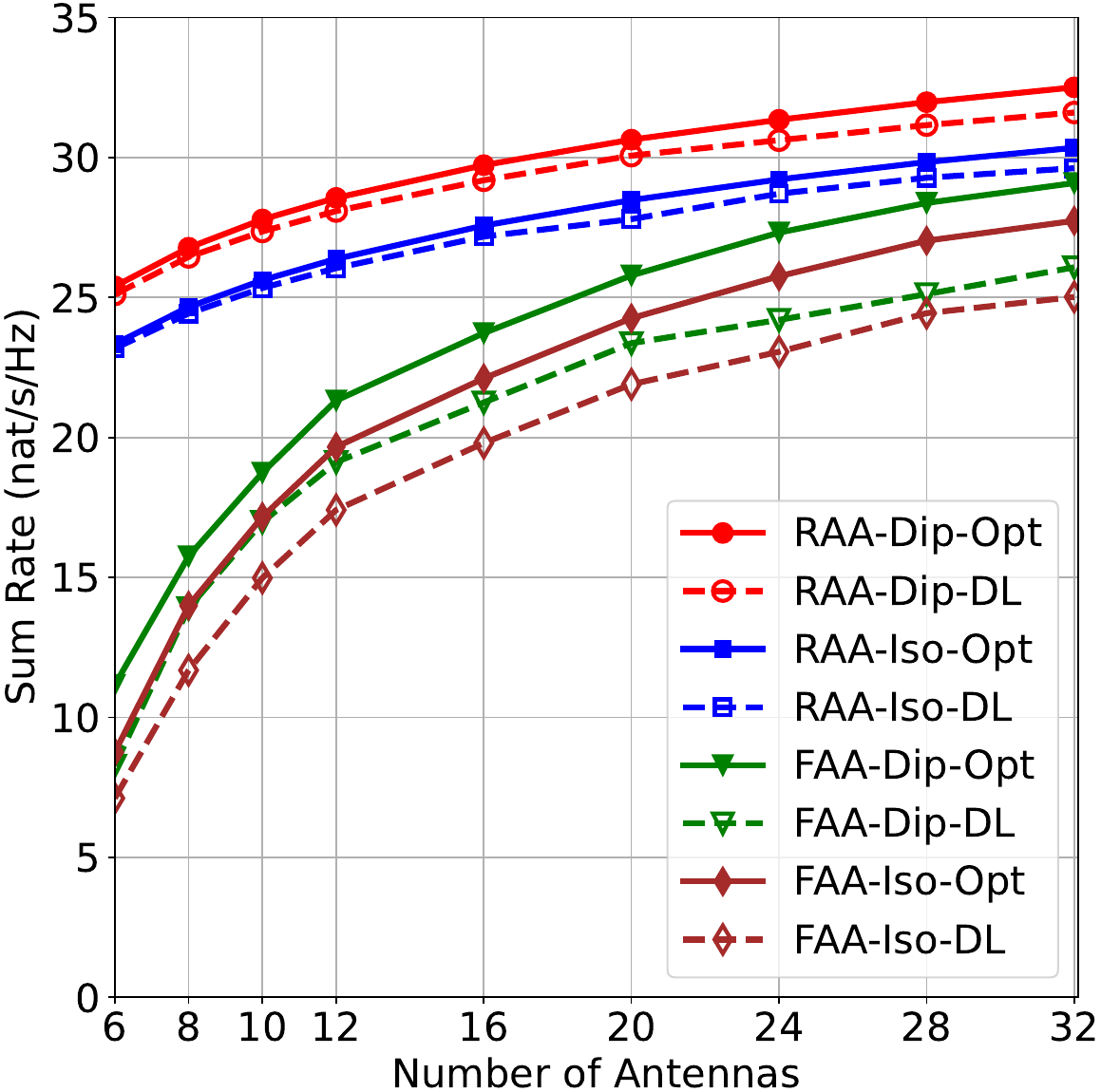}}
\subfigure[]{\label{FVSantenna}\includegraphics[width=0.49\columnwidth]{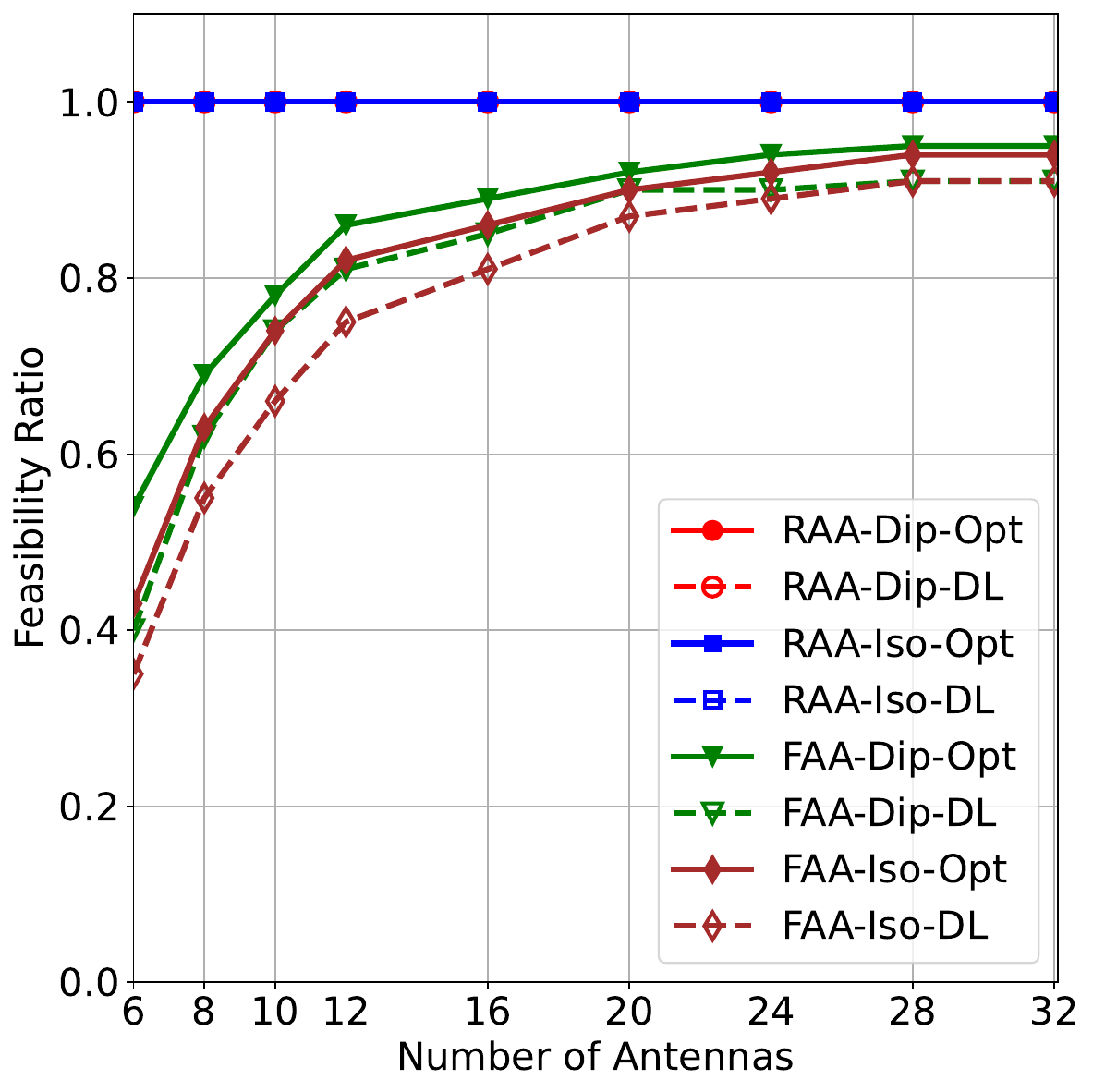}}
\vspace{-0.5cm}
\caption{(a) Sum-rate and (b) feasibility ratio versus the number of antennas.
}
\vspace{-0.2cm}
\label{increase_antenna}
\end{figure}

Next, we compare the sum-rates achieved by all considered schemes, averaged over 100 random realizations of the user locations. \emph{Note that the sum-rate obtained by a given scheme for a given realization is set to zero whenever the obtained solution fails to satisfy the QoS constraint in problem \eqref{P1}} to clearly illustrate the advantages of the proposed RAA-Dip-Opt and RAA-Dip-DL schemes. Meanwhile, we define the feasibility ratio of a given scheme as the percentage of feasible solutions attained over 100 random realizations to evaluate the QoS constraint satisfaction. Fig.~\ref{increase_QoS} depicts the average sum-rates and feasibility ratios for all considered optimization (in solid lines) and DL (in dashed lines) schemes versus the minimum required achievable rate $\overline{R}_k$ of each user. We observe that the average sum-rates and feasibility ratios for all schemes decrease monotonically with $\overline{R}_k$. This is expected because a higher $\overline{R}_k$ shrinks the feasible region of \eqref{P1} and reduces the flexibility for sum-rate maximization or even renders \eqref{P1} infeasible. However, the schemes employing RAAs significantly outperform the schemes based on FAAs, particularly when $\overline{R}_k$ exceeds 6 nat/s/Hz. For example, the average sum-rate of the proposed RAA-Dip-Opt remains at approximately 30 nat/s/Hz when $\overline{R}_k$ is below 6 nat/s/Hz and then decreases only slightly, whereas the average sum-rate of FAA-Dip-Opt decreases continuously and significantly from about 25 nat/s/Hz to 7 nat/s/Hz. This is because RAAs can proactively adjust the users' channels via orientation design, thereby enhancing the received signal power of the intended users, mitigating multiuser interference, and enlarging the feasible region of problem \eqref{P1} (cf. Fig.~\ref{FVSQoS}).

Fig.~\ref{increase_QoS} also reveals that the proposed optimization framework outperforms the proposed DL framework in terms of the sum-rate, especially when $\overline{R}_k$ is large. This is due to the fact that the optimization framework based on the PDD method treats each realization of the user positions individually and optimizes $\{\boldsymbol{W}$, $\boldsymbol{r}\!_{\text{a}}$, $\boldsymbol{r}\!_{\text{d}}\}$ iteratively with closed-form solutions for the involved subproblems, achieving high reliability even under stringent QoS requirements. In contrast, the DL framework learns the joint design rules via training over a batch of realizations. As a result, the DL framework maximizes the sum-rate and improves the QoS satisfaction on average over the batch, making it difficult to precisely satisfy stringent QoS requirements for individual realizations. However, with the two-module architecture and the two-stage training strategy, the DL framework approaches the sum-rate performance of the optimization framework when $\overline{R}_k$ is relatively small. For example, the average sum-rate of RAA-Dip-DL is only about $2\%$ lower than that of RAA-Dip-Opt when $\overline{R}_k$ lies between 1 nat/s/Hz and 5 nat/s/Hz. For stringent QoS requirements, RAA-Dip-DL still suffers less than $5\%$ sum-rate loss compared with RAA-Dip-Opt for the feasible solutions. The performance gap between the average sum-rates of RAA-Dip-DL and RAA-Dip-Opt for $\overline{R}_k>5$ nat/s/Hz is mainly due to the relatively lower feasibility ratio of the DL framework. Note that both RAA-Dip-Opt and RAA-Dip-DL yield higher average sum-rates than the baselines RAA-Iso-Opt and RAA-Iso-DL, particularly when $\overline{R}_k$ exceeds 6 nat/s/Hz, owing to the high directivity of dipole antennas.

Fig.~\ref{increase_antenna} presents the average sum-rates and feasibility ratios versus the number of transmit antennas $N$. We observe that the average sum-rates and feasibility ratios increase monotonically with $N$ for all considered schemes, as more spatial DoFs are exploitable for sum-rate maximization. Moreover, the schemes using RAAs outperform those using FAAs, particularly for a small $N$. For example, when $N=6$, the proposed RAA-Dip-Opt achieves an average sum-rate that is almost twice as high as that of FAA-Dip-Opt (cf. Fig.~\ref{VSantenna}). This is because FAAs with few antennas have limited capability to distinguish users in azimuth angles, making it difficult to mitigate multiuser interference and satisfy QoS constraints, thereby resulting in a much lower feasibility ratio (cf. Fig.~\ref{FVSantenna}) and average sum-rate. However, all schemes based on RAAs attain a feasibility ratio close to $100\%$ for $N$ ranging from 6 to 32, highlighting the potential of RAAs in systems requiring a small number of antennas, such as UAV platforms under SWAP constraints. Additionally, the high efficiency of the DL framework is evident in Fig.~\ref{VSantenna}, as the performance loss is less than $3\%$ when comparing RAA-Dip-DL with RAA-Dip-Opt.

\begin{figure}[t]
\centering
\subfigure[]{\label{VSuser}\includegraphics[width=0.49\columnwidth]{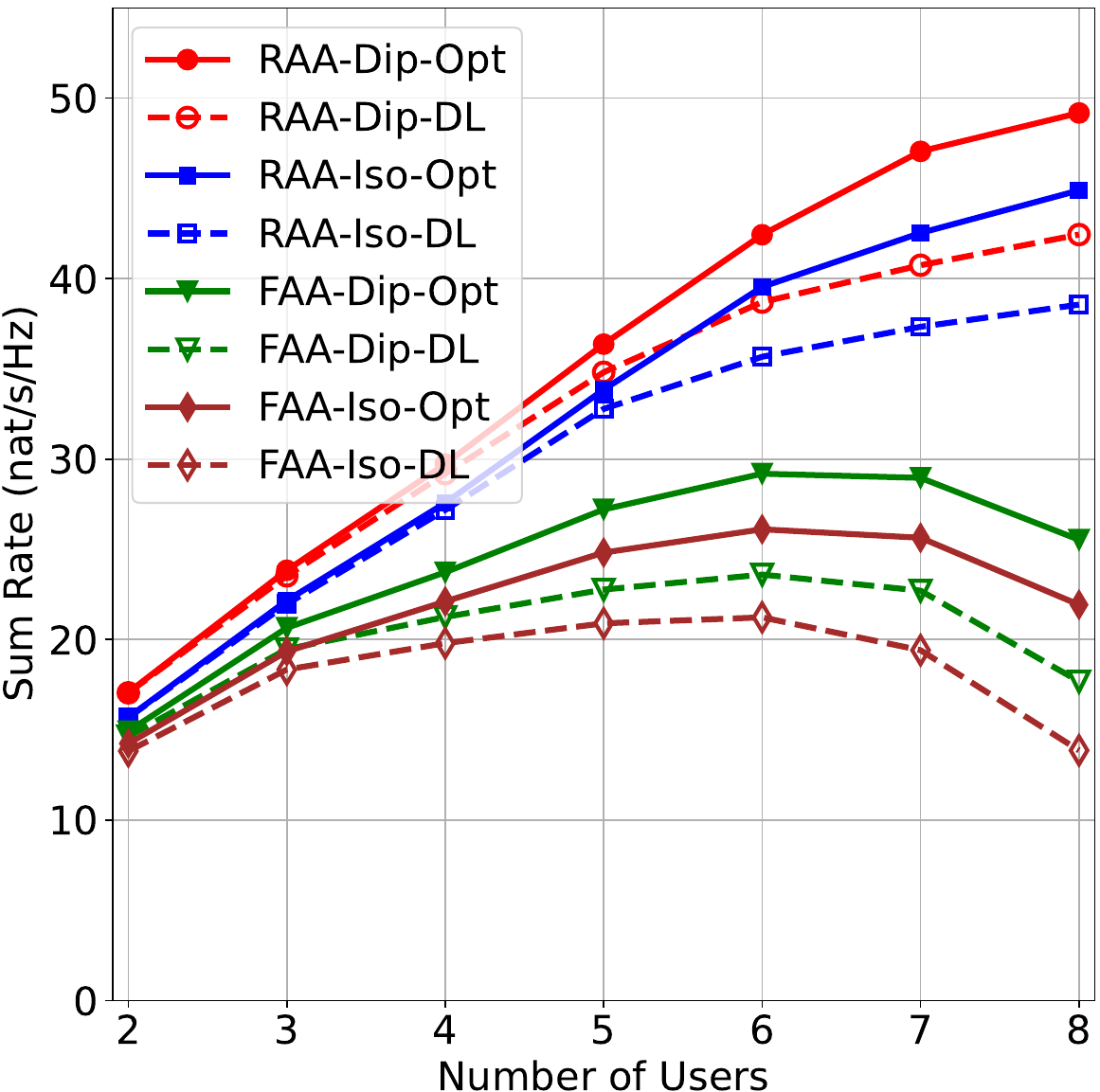}}
\subfigure[]{\label{FVSuser}\includegraphics[width=0.49\columnwidth]{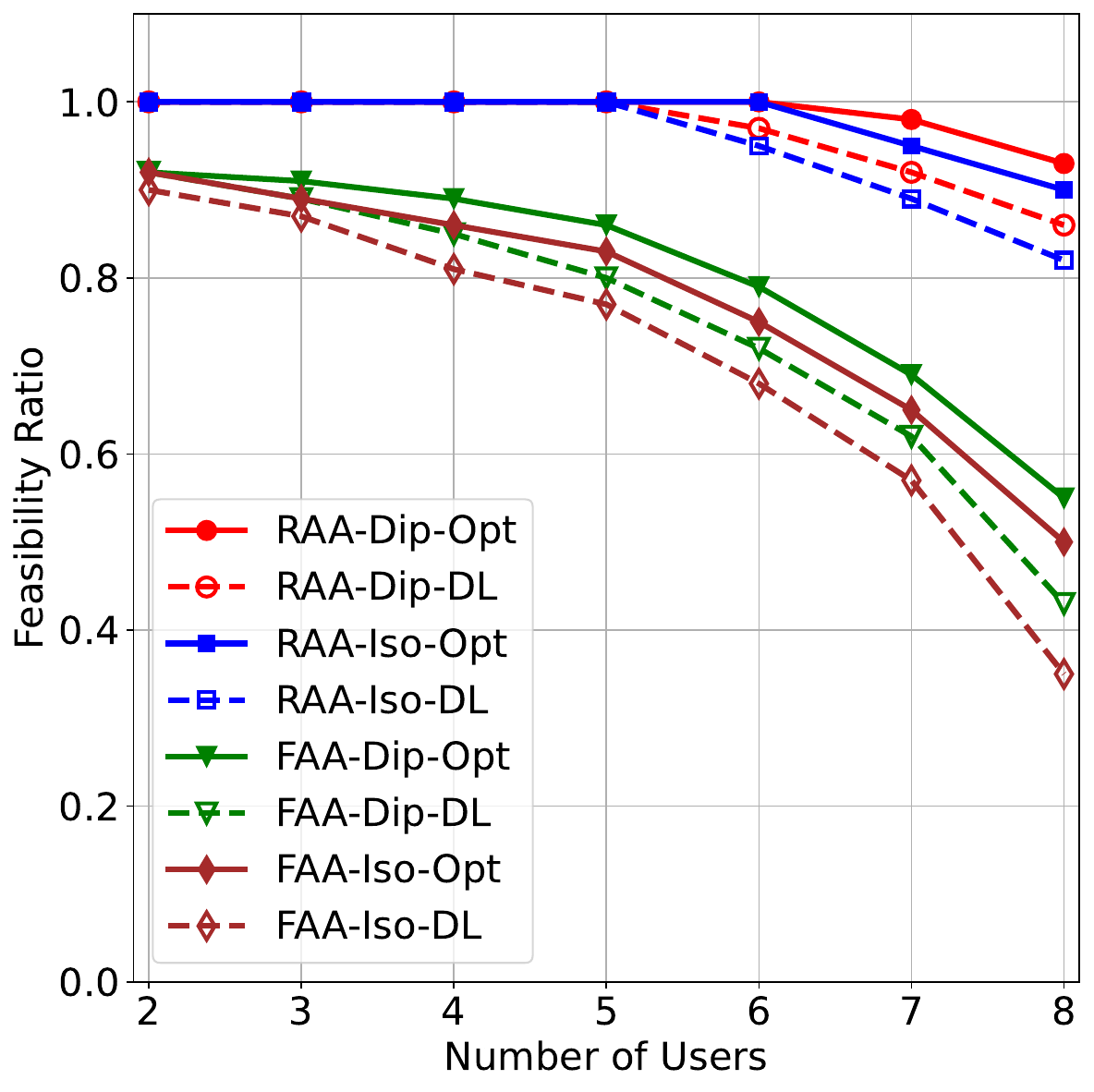}}
\vspace{-0.5cm}
\caption{(a) Sum-rate and (b) feasibility ratio versus the number of users.
}
\vspace{-0.3cm}
\label{increase_user}
\end{figure}
\begin{figure}[t]
\centering \includegraphics[width=0.99\columnwidth]{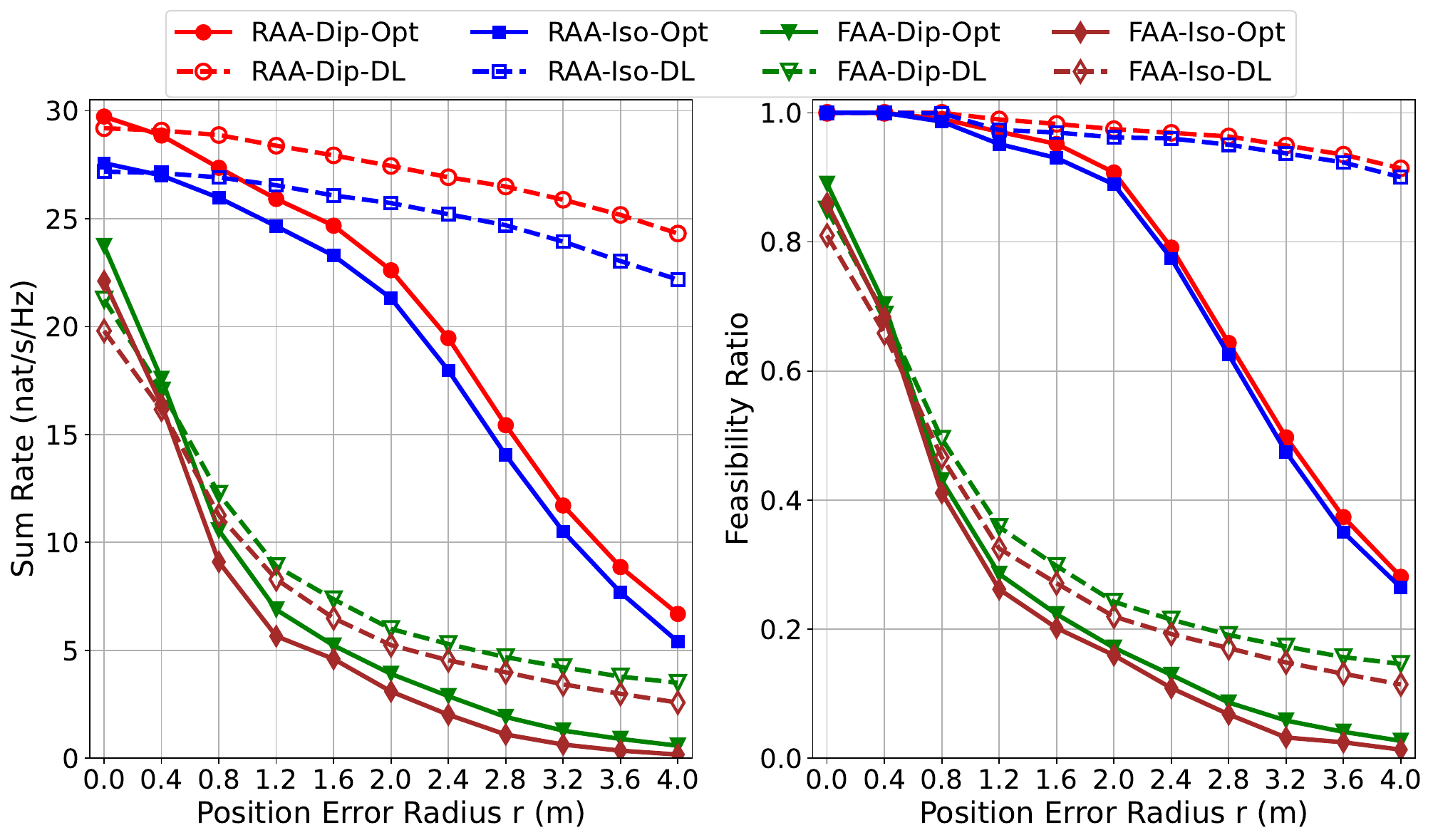} 
\vspace{-0.5cm}
\caption{Sum-rate and feasibility ratio versus radius $r$ (m) of position error.}
\vspace{-0.1cm}
\label{error}
\end{figure}
Fig.~\ref{increase_user} further compares the average sum-rates and feasibility ratios versus the number of users $K$. From Fig.~\ref{VSuser}, we observe that the average sum-rates for all schemes increase with $K$ for $K\leq6$. This is because, when $K$ is small, the QoS requirements can be satisfied for both FAAs and RAAs using only part of the available resources. The remaining resources are allocated to users with favorable channels, thereby exploiting the multiuser diversity gain for higher sum-rates as $K$ increases. However, for $K>6$, the average sum-rates of the FAA-based schemes decrease, whereas those of the RAA-based schemes continue to increase. This is because, when more users are present in a given area, multiuser interference becomes severe for FAAs. As a result, FAA-based schemes will spend almost all resources satisfying the stringent QoS requirements, and in severe cases even fail to meet them, leaving no resources for exploiting the multiuser diversity gain. By contrast, RAAs can still effectively mitigate multiuser interference for large $K$ by orientation design, and retain sufficient resources to exploit the multiuser diversity gain, thereby achieving much higher average sum-rates and feasibility ratios (cf. Fig.~\ref{FVSuser}). This highlights the advantage of RAAs in simultaneously satisfying QoS constraints and exploiting multiuser diversity.
\vspace{-0.4cm}
\subsection{Robustness Against Imperfect User Positions}
Finally, we evaluate the robustness of all considered schemes under imperfect user position information. We adopt a practical uniformly bounded error model \cite{r11,r35}, where the estimated user positions are uniformly distributed within a disk of radius $r$ (m) centered at the actual user location. In the simulation, the actual user locations are only known for performance evaluation, but are unknown to any scheme during the joint design. As such, all schemes operate directly based on the estimated positions with errors. Fig.~\ref{error} depicts the average sum-rates and feasibility ratios versus radius $r$. We observe that, as expected, the average sum-rates and feasibility ratios decrease with $r$, reflecting the importance of accurate user position information for high performance UAV-aided communication. However, for the RAA-based schemes, the average sum-rates degrade much slower with $r$ than their FAA-based counterparts. This is because RAAs can separate users in the azimuth angles, such that the resulting sum-rates are less sensitive to small position perturbations. In contrast, as can be observed from Fig.~\ref{FAAGain}, when employing FAA-Dip-Opt, slight position perturbations of user 1 and user 3 can dramatically reduce the sum-rate or even violate the QoS constraints. Interestingly, Fig.~\ref{error} demonstrates that the proposed DL framework is much more robust to imperfect user position information than the proposed optimization framework. This is because the DL framework learns the input-output mapping from the user positions to the predicted $\{\boldsymbol{W}$, $\boldsymbol{r}\!_{\text{a}}$, $\boldsymbol{r}\!_{\text{d}}\}$ over a large number of training samples, whereas the optimization framework is specifically tailored to each realization and is therefore more sensitive to position errors. 
\vspace{-0.3cm}
\section{Conclusions}
In this paper, we investigated RAA-based UAV-aided communication. To maximize the sum-rate subject to per-user QoS requirements, we jointly optimized the RAA orientation and beamforming. To solve the formulated highly nonconvex problem, we first proposed an optimization framework based on the PDD method, which can efficiently handle the coupling between beamforming and orientation design while guaranteeing convergence to a KKT solution. We further proposed a GNN-based DL framework which comprises two serially connected modules and a two-stage training strategy to overcome the gradient imbalance for joint beamforming and orientation design during training. Simulation results demonstrated that RAAs significantly outperform FAAs in terms of mitigating interference and exploiting multiuser diversity, achieving significantly higher sum-rates under QoS constraints. Moreover, the proposed optimization framework was shown to be capable of satisfying stringent QoS requirements. In comparison, the proposed DL framework could significantly reduce the computation time while approaching the performance of the optimization framework, and exhibited great robustness against imperfect user position information.
\appendix[Proof of Theorem 2]
The Lagrangian associated with problem \eqref{P6} is defined as
\begin{align}
&\mathcal{L}(\boldsymbol{x}_k,\lambda_{2,k})=f(\boldsymbol{x}_k,\boldsymbol{x}_k^{(i)})\notag \\
&+\tfrac{1}{2\rho_1}\sum\nolimits_{m=1}^{K}|x_{k,m}-\boldsymbol{h}_k^{\mathrm{H}}\boldsymbol{w}_m+\rho_1\boldsymbol{Z}_{1(k,m)}|^2 \notag \\
&+\lambda_{2,k}\overline{\gamma}_k\Big(\sum\nolimits_{m=1, m \neq k}^{K}|x_{k,m}|^2+\sigma^{2}_k\Big)-\lambda_{2,k}|x_{k,k}|^2.
\end{align}
The optimal $\boldsymbol{x}_k^\star$ and $\lambda_{2,k}^\star$ have to satisfy the KKT conditions, including the stationarity conditions given by 
\begin{align}
&\left.\partial \mathcal{L}(\boldsymbol{x}_k,\lambda_{2,k})\Big/\partial x_{k,k}^*\right|_{x_{k,k}=x_{k,k}^{\star}}=(\delta(\boldsymbol{x}_k^{(i)})-\lambda_{2,k}^\star)x_{k,k}^\star \notag \\
&+\tfrac{1}{2\rho_1}(\rho_1\boldsymbol{Z}_{1(k,k)}-\boldsymbol{h}_k^\mathrm{H}\boldsymbol{w}_k)-x_{k,k}^{(i)}/\zeta_k(\boldsymbol{x}_k^{(i)})=0, 
\label{station1} \\
&\left.\partial \mathcal{L}(\boldsymbol{x}_k,\lambda_{2,k})\Big/\partial x_{k,m}^*\right|_{x_{k,m}=x_{k,m}^{\star}}=(\delta(\boldsymbol{x}_k^{(i)})+\lambda_{2,k}^\star\overline{\gamma}_k)x_{k,m}^\star \notag \\
&+\tfrac{1}{2\rho_1}(\rho_1\boldsymbol{Z}_{1(k,m)}-\boldsymbol{h}_k^\mathrm{H}\boldsymbol{w}_m)=0, \ m\neq k,
\label{station2}
\end{align}
the complementary slackness condition
\begin{align}
\lambda_{2,k}^\star\overline{\gamma}_k\big(\sum\nolimits_{m=1, m \neq k}^{K}|x_{k,m}^\star|^2+\sigma^{2}_k\big)-\lambda_{2,k}^\star|x_{k,k}^\star|^2=0,
\label{CSC}
\end{align}
the dual feasibility condition $\lambda_{2,k}^\star\geq 0$, and the primal feasibility condition
\begin{align}
\overline{\gamma}_k\Big(\sum\nolimits_{m=1, m \neq k}^{K}|x_{k,m}^\star|^2+\sigma^{2}_k\Big)-|x_{k,k}^\star|^2 \leq 0. 
\label{PFC}
\end{align}
Note that \eqref{xkm} directly follows from \eqref{station2}. 

Now, let us assume that $\frac{1}{2\rho_1}(\rho_1\boldsymbol{Z}_{1(k,k)}-\boldsymbol{h}_k^\mathrm{H}\boldsymbol{w}_k)-x_{k,k}^{(i)}/\zeta_k(\boldsymbol{x}_k^{(i)})=0$. As $|x_{k,k}^\star|^2>0$ in \eqref{PFC}, the optimal Lagrangian multiplier $\lambda_{2,k}^\star=\delta(\boldsymbol{x}_k^{(i)})$ and \eqref{xkk1} can be derived from the conditions \eqref{station1} and \eqref{CSC}, respectively. 

Subsequently, assume $\frac{1}{2\rho_1}(\rho_1\boldsymbol{Z}_{1(k,k)}\!-\!\boldsymbol{h}_k^\mathrm{H}\boldsymbol{w}_k)\!-\!x_{k,k}^{(i)}/\zeta_k(\boldsymbol{x}_k^{(i)})\!\neq\! 0$. The proof of \eqref{xkk2} follows from \eqref{station1}.

Finally, by substituting \eqref{xkm} and \eqref{xkk2} into \eqref{CSC}, we have $\lambda_{2,k}^\star\kappa(\lambda_{2,k}^\star)=0$, where
\begin{align}
&\kappa(\lambda_{2,k}^\star)=\overline{\gamma}_k\Big(\sum\nolimits_{m=1, m \neq k}^{K}\frac{|\boldsymbol{h}_k^\mathrm{H}\boldsymbol{w}_m-\rho_1\boldsymbol{Z}_{1(k,m)}|^2}{4\rho_1^2(\delta(\boldsymbol{x}_k^{(i)})+\lambda_{2,k}^\star\overline{\gamma}_k)^2}+\sigma^{2}_k\Big) \notag \\
&-\frac{|\frac{1}{2\rho_1}(\rho_1\boldsymbol{Z}_{1(k,k)}-\boldsymbol{h}_k^\mathrm{H}\boldsymbol{w}_k)-x_{k,k}^{(i)}/\zeta_k(\boldsymbol{x}_k^{(i)})|^2}{(\delta(\boldsymbol{x}_k^{(i)})-\lambda_{2,k}^\star)^2},
\end{align}
and $\kappa(\lambda_{2,k}^\star)$ monotonically decreases w.r.t. $\lambda_{2,k}^\star$. Since $\lambda_{2,k}^\star\geq 0$, this leads to $\lambda_{2,k}^\star=0$ when $\kappa(0)\leq 0$; otherwise, $\lambda_{2,k}^\star$ can be found by solving $\lambda_{2,k}^\star\kappa(\lambda_{2,k}^\star)=0$ via bisection search. This completes the proof.


\vspace{-0.4cm}
\begin{thebibliography}{1}
\bibliographystyle{IEEEtran}

\bibitem{r1} Z. Xiao et al., ``A Survey on Millimeter-Wave Beamforming
Enabled UAV Communications and Networking," \emph{IEEE Commun.
Surveys Tuts.}, vol. 24, no. 1, pp. 557-610, 2022.

\bibitem{r2} M. Mozaffari et al., ``A Tutorial on UAVs for Wireless
Networks: Applications, Challenges, and Open Problems," \emph{IEEE Commun.
Surveys Tuts.}, vol. 21, no. 3, pp. 2334-2360,
2019.

\bibitem{r3}
T. L. Marzetta, ``Noncooperative Cellular Wireless with Unlimited Numbers of Base Station Antennas,'' \emph{IEEE Trans. Wireless Commun.}, vol. 9, no. 11, pp. 3590-3600, Nov. 2010.

\bibitem{r4}
Q. Wu, et al., ``A Comprehensive Overview on 5G-and-beyond Networks with UAVs,'' \emph{IEEE J. Sel. Areas Commun.}, vol. 39, no. 10, pp. 2912-2945, Oct. 2021.

\bibitem{r5}
Z. Lyu, G. Zhu and J. Xu, ``Joint Maneuver and Beamforming Design for UAV-enabled Integrated Sensing and Communication,'' \emph{IEEE Trans. Wireless Commun.}, vol. 22, no. 4, pp. 2424-2440, Apr. 2023.

\bibitem{r6}
Y. K. Tun et al., ``Joint Beamforming and Trajectory Optimization for Multi-UAV-Assisted Integrated Sensing and Communication Systems,'' \emph{IEEE Trans. Veh. Technol.}

\bibitem{r10}
B. Zheng et al., ``Rotatable Antenna Enabled Wireless Communication and Sensing: Opportunities and Challenges," \emph{IEEE Wireless Commun.}, vol. 33, no. 3, pp. 134-141, Jun. 2026.

\bibitem{r7}
L. Zhu et al., ``Movable Antennas for Wireless Communication: Opportunities and Challenges," \emph{IEEE Commun. Mag.}, vol. 62, no. 6, pp. 114-120, June 2024.

\bibitem{r8}
B. Zheng et al., ``Rotatable Antenna-Enabled Wireless Communication: Modeling and Optimization," \emph{IEEE Trans. Wireless Commun.}, vol. 74, pp. 6825-6842, 2026.

\bibitem{6DMA}
X. Shao et al., ``6D Movable Antenna Based on User Distribution: Modeling and Optimization," \emph{IEEE Trans.
Wireless Commun.}, vol. 24, no. 1, pp. 355-370, Jan. 2025.

\bibitem{r11}
X. Zhang et al., ``Rotatable Antenna Array Enabled UAV mmWave Massive MIMO Communication," \emph{IEEE Trans.
Wireless Commun.}, vol. 74, pp. 1219-1236, 2026.

\bibitem{r14} 
L. Xiang et al., ``Joint Optimization of Beamforming
and 3D Array Steering for Multi-Antenna UAV Communications," \emph{
IEEE WCNC}, Dubai, 2024.

\bibitem{r12} L. Xiang et al., ``Energy-Efficient Dynamic Array-Steering
and Beamforming for UAV-Aided Communications”, \emph{IEEE GlobeCom},
Cape Town, 2024.

\bibitem{r13}
B. Yilmaz et al., ``Completion Time Minimization for UAV-aided Communications with Rotatable Dipole Array”, \emph{IEEE ICC},
Montreal, 2025.

\bibitem{r15} F. Pei et al., ``Joint Optimization of Beamforming
and 3D Array-Steering for UAV-Aided ISAC," \emph{IEEE ICC}, Denver,
CO, USA, 2024.

\bibitem{r16} F. Pei et al., ``Transmit Beamforming and Array Steering
Optimization for UAV-Aided Bistatic ISAC”, \emph{IEEE GlobeCom}, Cape
Town, 2024.

\bibitem{r25}
Q. Shi et al., ``Penalty Dual Decomposition Method for Nonsmooth Nonconvex Optimization—Part I: Algorithms and Convergence Analysis," \emph{IEEE Trans. Signal Process.}, vol. 68, pp. 4108-4122, 2020.

\bibitem{r26}
Q. Shi et al., ``Penalty Dual Decomposition Method for Nonsmooth Nonconvex Optimization—Part II: Applications," \emph{IEEE Trans. Signal Process.}, vol. 68, pp. 4242-4257, 2020.

\bibitem{r27}
Q. Shi et al., ``Spectral Efficiency Optimization For Millimeter Wave Multiuser MIMO Systems," \emph{IEEE J. Select. Topics Signal Process.}, vol. 12, no. 3, pp. 455-468, Jun. 2018.

\bibitem{r17}
Y. Shen et al., ``Graph Neural Networks for Scalable Radio Resource Management: Architecture Design and Theoretical Analysis," \emph{IEEE J. Select. Areas Commun.}, vol. 39, no. 1, pp. 101-115, Jan. 2021.

\bibitem{r18}
Y. Shen et al., ``Graph Neural Networks for Wireless Communications: From Theory to Practice," \emph{IEEE Trans.
Wireless Commun.}, vol. 22, no. 5, pp. 3554-3569, May 2023.

\bibitem{r19}
Y. Li et al., ``GNN-Based Beamforming for Sum-Rate Maximization in MU-MISO Networks," \emph{IEEE Trans.
Wireless Commun.}, vol. 23, no. 8, pp. 9251-9264, Aug. 2024.

\bibitem{r20}
N. Shlezinger et al., ``DeepSIC: Deep Soft Interference Cancellation for Multiuser MIMO Detection," \emph{IEEE Trans. Wireless Commun.}, vol. 20, no. 2, pp. 1349-1362, Feb. 2021.

\bibitem{r21}
A. Beyazıt et al., ``On Stream Selection for Interference Alignment in Heterogeneous Networks," \emph{J. Wireless Commun. Network}, 80 (2016).

\bibitem{r22} C. A. Balanis, \emph{Antenna Theory: Analysis and Design}. 4th ed, 2016.

\bibitem{r23} D. Xu et al., ``Multiuser MISO UAV Communications in Uncertain Environments With No-Fly Zones: Robust Trajectory and Resource Allocation Design," \emph{IEEE Trans. Commun.}, vol. 68, no. 5, pp. 3153-3172, May 2020.

\bibitem{r24}
X. Dai et al., ``Energy-Efficient UAV Communications in the Presence of Wind: 3D Modeling and Trajectory Design," \emph{IEEE Trans.
Wireless Commun.}, vol. 23, no. 3, pp. 1840-1854, March 2024.

\bibitem{r28}
M.R.Hestenes, ``Multiplier and Gradient Methods,” \emph{J. Optim. Theory Appl.}, vol. 4, no. 5, pp. 303–320, 1969.

\bibitem{r29}
S. Boyd et al., \emph{Convex Optimization}, Cambridge University Press, 2004.

\bibitem{r30}
X. Zhang et al., ``Massive MIMO Multicasting With Finite Blocklength," \emph{IEEE Trans.
Wireless Commun.}, vol. 23, no. 10, pp. 15018-15034, Oct. 2024.

\bibitem{r31}
Y. Sun et al., ``Majorization-Minimization Algorithms in Signal Processing, Communications, and Machine Learning," \emph{IEEE Trans. Signal Process.}, vol. 65, no. 3, pp. 794-816, Feb.1, 2017.

\bibitem{r32} P.-A. Absil et al., \emph{Optimization Algorithms on
Matrix Manifolds.} Princeton University Press, 2009.

\bibitem{r33}
J. Nocedal et al., \emph{Numerical Optimization}, 2nd Edition, Springer, 2006.

\bibitem{r34}
S. K. Joshi et al., ``Weighted Sum-Rate Maximization for MISO Downlink Cellular Networks via Branch and Bound," \emph{IEEE Trans. Signal Process.}, vol. 60, no. 4, pp. 2090-2095, April 2012.

\bibitem{GNN}
W. L. Hamilton, \emph{Graph Representation Learning}, Morgan \& Claypool, 2020.

\bibitem{r35}
R. Ismayilov et al., ``Adaptive Beam-Frequency Allocation Algorithm with Position Uncertainty for Millimeter-Wave MIMO Systems," \emph{IEEE VTC Spring}, Porto, Portugal, 2018.

\end{thebibliography}
\end{document}